\documentclass[aps,pra,twocolumn,nofootinbib, superscriptaddress,10pt]{revtex4-2}

\usepackage{graphicx}
\usepackage{epstopdf}
\usepackage{amsmath}
\usepackage{amssymb}
\usepackage{mathrsfs}
\usepackage{amsthm}
\usepackage{bm}
\usepackage{url}
\usepackage[T1]{fontenc}
\usepackage{csquotes}
\MakeOuterQuote{"}
\usepackage{thmtools, thm-restate}

\usepackage{chngcntr} 

\newtheoremstyle{note}
{\topsep/2}               
{\topsep/2}               
{}                      
{\parindent}            
{\itshape}              
{.}                     
{5pt plus 1pt minus 1pt}
{}

\theoremstyle{note}
\newtheorem{thm}{Theorem}
\newtheorem{lem}{Lemma}

\newtheorem{corollary}{Corollary}
\newtheorem{proposition}{Proposition}

\theoremstyle{definition}

\theoremstyle{remark}

\newtheorem{thm*}{Theorem}

\makeatletter
\expandafter\newcommand\csname thethm*default\endcsname{\thethm*}
\newcommand{\thmstarnum}[1]{\expandafter\gdef\csname thethm*\endcsname{#1*}}
\thmstarnum{\thethm}
\expandafter\g@addto@macro\csname endthm*\endcsname{\thmstarnum{\thethm}}
\makeatother

\newtheorem{lem*}{Lemma}

\makeatletter
\expandafter\newcommand\csname thelem*default\endcsname{\thelem*}
\newcommand{\lemstarnum}[1]{\expandafter\gdef\csname thelem*\endcsname{#1*}}
\lemstarnum{\thelem}
\expandafter\g@addto@macro\csname endlem*\endcsname{\lemstarnum{\thelem}}
\makeatother

\def\vec#1{\bm{#1}} 

\newcommand{\tr}{\operatorname{tr}}

\newcommand{\rk}{\operatorname{rank}}
\newcommand{\spa}{\operatorname{span}}

\newcommand{\imply}{\mathrel{\Rightarrow}}

\newcommand{\rmw}{\mathrm{w}}

\newcommand{\rmL}{\mathrm{L}}

\newcommand{\rmR}{\mathrm{R}}
\newcommand{\rmS}{\mathrm{S}}
\newcommand{\rmT}{\mathrm{T}}

\newcommand{\eig}{\mathrm{eig}}

\newcommand{\caF}{\mathcal{F}}
\newcommand{\caH}{\mathcal{H}}

\newcommand{\caJ}{\mathcal{J}}

\newcommand{\caL}{\mathcal{L}}

\newcommand{\caR}{\mathcal{R}}
\newcommand{\caS}{\mathcal{S}}

\newcommand{\caV}{\mathcal{V}}

\newcommand{\bbR}{\mathbb{R}}

\newcommand{\supp}{\mathrm{supp}}

\newcommand{\scrA}{\mathscr{A}}
\newcommand{\scrB}{\mathscr{B}}

\newcommand{\scrD}{\mathscr{D}}

\newcommand{\scrP}{\mathscr{P}}

\newcommand{\lsp}{\hspace{0.1em}}

\newcommand{\be}{\begin{equation}}
\newcommand{\ee}{\end{equation}}
\newcommand{\ba}{\begin{align}}
\newcommand{\ea}{\end{align}}

\def\<{\langle}  
\def\>{\rangle}  

\newcommand{\dket}[1]{| #1\>\!\>}

\newcommand{\dbra}[1]{\<\!\< #1|}

\newcommand{\dinner}[2]{\<\!\< #1| #2\>\!\>}

\newcommand{\douter}[2]{| #1\>\!\>\<\!\< #2|}

\newcommand{\range}{\operatorname{range}}

\newcommand{\bid}{\bar{\mathbf{I}}}

\def\eqref#1{\textup{(\ref{#1})}}  
\newcommand{\eref}[1]{Eq.~\textup{(\ref{#1})}}
\newcommand{\Eref}[1]{Equation~\textup{(\ref{#1})}}

\newcommand{\esref}[2]{Eqs.~\textup{(\ref{#1})} and \textup{(\ref{#2})}}
\newcommand{\essref}[3]{Eqs.~\textup{(\ref{#1})}, \textup{(\ref{#2})}, and \textup{(\ref{#3})}}
\newcommand{\Esref}[2]{Equations~\textup{(\ref{#1})} and \textup{(\ref{#2})}}

\newcommand{\fref}[1]{Fig.~\ref{#1}}

\newcommand{\thref}[1]{Theorem~\ref{#1}}
\newcommand{\Thref}[1]{Theorem~\ref{#1}}
\newcommand{\thsref}[1]{Theorems~\ref{#1}}
\newcommand{\Thsref}[1]{Theorems~\ref{#1}}

\newcommand{\lref}[1]{Lemma~\ref{#1}}
\newcommand{\Lref}[1]{Lemma~\ref{#1}}
\newcommand{\lsref}[1]{Lemmas~\ref{#1}}
\newcommand{\Lsref}[1]{Lemmas~\ref{#1}}

\newcommand{\pref}[1]{Proposition~\ref{#1}}
\newcommand{\Pref}[1]{Proposition~\ref{#1}}

\newcommand{\crref}[1]{Corollary~\ref{#1}}

\newcommand{\cref}[1]{Conjecture~\ref{#1}}
\newcommand{\Cref}[1]{Conjecture~\ref{#1}}

\newcommand{\aref}[1]{Appendix~\ref{#1}}

\newcommand{\rcite}[1]{Ref.~\cite{#1}}

\begin{document}
\title{Information Theoretic Significance of Projective Measurements}

\author{Huangjun Zhu}

\email{zhuhuangjun@fudan.edu.cn}

\affiliation{State Key Laboratory of Surface Physics and Department of Physics, Fudan University, Shanghai 200433, China}

\affiliation{Institute for Nanoelectronic Devices and Quantum Computing, Fudan University, Shanghai 200433, China}

\affiliation{Center for Field Theory and Particle Physics, Fudan University, Shanghai 200433, China}

\begin{abstract}
Projective measurements in quantum theory have  a very simple algebraic definition, but their information theoretic significance is quite elusive. 
Here we introduce a simple order relation based on the concentration of Fisher information, which complements the familiar data-processing order. Under this order relation, the information theoretic significance of projective measurements stands out immediately. Notably, projective measurements are exactly those quantum measurements whose extracted Fisher information is as concentrated as possible, which we call Fisher-sharp measurements. We also introduce the concept of sharpness index and show that  it is completely determined by the finest projective measurement
among the coarse graining of a given measurement.  
\end{abstract}

\date{\today}
\maketitle

\emph{Introduction.}---Quantum measurements play a key role in quantum theory and in various tasks in quantum information processing \cite{Neum55,NielC10book,BuscLPY16}. Projective measurements, also known as sharp measurements, are the most fundamental quantum measurements  as discussed in most elementary textbooks. Although they have a very simple algebraic definition, their information theoretic significance is still quite elusive despite the intensive efforts of numerous researchers in the past century. For example, what is sharp in a projective measurement? Does it has a simple information theoretic interpretation.

Here we introduce a simple order relation on quantum measurements based on the concentration of Fisher information.  This order relation can manifest important features that are ignored in the familiar data-processing order \cite{MartM90,Zhu15IC,ZhuHC16,Zhu22,HeinJN22} and is  quite useful to understanding elementary quantum measurements. We then introduce the concept of  Fisher-sharp measurements as  those quantum measurements whose extracted Fisher information is as concentrated as possible. Surprisingly, a quantum measurement is Fisher sharp iff it is a projective measurement. This result establishes a satisfactory connection between the algebraic concept of sharp measurements and the information theoretic concept of Fisher sharp measurements and thereby  endow projective measurements with a simple information theoretic interpretation, which is missing for too long a time. Furthermore, we  introduce the  sharpness index and show that it is completely determined by the   projective index, which characterizes the finest projective measurement
among the coarse graining of a given measurement.  In our study, we derive a number of results on quantum estimation theory, which are of interest to foundational studies and practical applications, such as multiparameter quantum metrology.

\emph{Quantum measurements.}---Let $\caH$ be a  $d$-dimensional Hilbert space. Quantum states on $\caH$ are  represented by density operators, which are positive (semidefinite) operators of trace~1. Here we are interested in the set of positive definite density operators, which is denoted by $\scrD(\caH)$ henceforth. A general quantum measurement on $\caH$ is described by a positive operator-valued measure (POVM)~\cite{NielC10book}, which is composed of a family of positive operators that sum up to the identity operator, denoted by~1 henceforth. 
A POVM is rank~1 if all nonzero POVM elements are rank 1. A projective measurement is described by a projector-valued measure (PVM), which is a special POVM composed of projectors.


The set of POVMs has a natural order relation induced by data or information processing \cite{MartM90,Zhu15IC,ZhuHC16}. 
Given two POVMs  $\scrA=\{A_j\}_j$ and $\scrB=\{B_k\}_k$  on $\caH$,  $\scrA$ is a \emph{coarse graining} of $\scrB$, expressed as $\scrA\leq  \scrB$,  if  $A_j=\sum_k\Lambda_{jk}B_k$
with  a stochastic matrix $\Lambda$. 
 A \emph{grouping} of $\scrB$ is a special coarse graining in which all nonzero entries of $\Lambda$ equal~1. The POVMs $\scrA, \scrB$  are \emph{equivalent} if they are coarse graining of each other. A POVM  is \emph{simple} if every pair of POVM elements is linearly independent, so  no zero POVM element is allowed. Every POVM is equivalent to a simple POVM, and two simple POVMs are equivalent iff their  POVM elements are identical up to relabeling \cite{MartM90,Kura15,Zhu22}. So we can focus on simple POVMs.

The \emph{transition graph} of a POVM $\scrA$, denoted by $G(\scrA)$, is a graph whose vertices are in one-to-one correspondence with nonzero POVM elements in $\scrA$; two vertices are adjacent iff the corresponding POVM elements are not orthogonal. $\scrA$ is \emph{irreducible} if $G(\scrA)$ is connected and \emph{reducible} otherwise. The POVM elements tied to each connected component of $G(\scrA)$ form an \emph{irreducible component} of $\scrA$. The \emph{projective index} of $\scrA$, denoted by $\gamma(\scrA)$, is defined as the number of irreducible components, which is equal to the number of nonzero POVM elements for a PVM. Let $\scrA_j$ for $j=1,2,\ldots, \gamma(\scrA)$ be the irreducible components of $\scrA$. 
Define $P_j:=\sum_{A\in \scrA_j} A$; then 
$\scrP(\scrA):=\{P_j\}_{j=1}^{\gamma(\scrA)}$ is the finest PVM among the coarse graining of $\scrA$ by \pref{pro:P(A)} below.  All proofs are relegated to the Supplemental Material. 
\begin{proposition}\label{pro:P(A)}
	Let  $\scrA$ be a POVM on $\caH$. Then a PVM is a coarse graining of $\scrA$ iff it is a grouping of $\scrP(\scrA)$. In addition, $\gamma(\scrA)=\gamma(\scrP(\scrA))\leq d$, and the upper bound is saturated iff $\scrA$ is equivalent to a rank-1 PVM. 
\end{proposition}

\emph{Quantum estimation theory.}---Suppose $\rho(\theta)\in \scrD(\caH)$  is specified by $g=d^2-1$ parameters $\theta_1, \theta_2, \ldots, \theta_g$.  If the  POVM $\scrA=\{A_j\}_{j=1}^m$ is performed, then the
probability of obtaining outcome $j$ is
$p_j(\theta)=\tr[\rho(\theta)A_j]$. The resulting \emph{Fisher
	information matrix} $I(\theta)$ has entries \cite{Fish25}
\begin{equation}
I_{ab}(\theta,\scrA)=\sum_{j}\frac{1}{p_j}\frac{\partial p_j}{\partial \theta_a}\frac{\partial p_j}{\partial \theta_b},\quad a,b=1,2,\ldots, g.
\end{equation}
The importance of $I(\theta,\scrA)$ is embodied  in the famous Cram\'er-Rao bound \cite{Rao02book}: the MSE matrix  (or covariance matrix) of any unbiased estimator is bounded from below by the inverse Fisher information matrix.

The \emph{quantum Fisher information  matrix}  $J_\rmS(\theta)$ \cite{Hels76book, Hole82book,BrauC94,Haya05book,LiuYLW19} has matrix entries
\begin{equation}
(J_\rmS)_{ab}(\theta)=\frac{1}{2}\tr\bigl [\rho(L_a L_b+L_bL_a )\bigr ],
\end{equation}
where  $L_a $  is the Hermitian operator determined by  the equation $\frac{1}{2}(\rho
L_a+L_a\rho)=\rho_{,a}:=\partial
\rho(\theta)/\partial \theta_a$ and is called  the \emph{symmetric logarithmic derivative} (SLD) associated with $\theta_a$ (assuming that $\rho_{,a}\neq 0$). The significance of $J_\rmS(\theta)$ lies in the SLD bound
\begin{align}\label{eq:QCR}
I(\theta,\scrA)\leq J_\rmS(\theta), 
\end{align}
which implies the quantum Cram\'er-Rao bound: 
$J_\rmS^{-1}(\theta)$ is a lower bound for the  MSE matrix of any unbiased estimator \cite{Hels76book, Hole82book,BrauC94,Haya05book,LiuYLW19}. This bound can be saturated in the one-parameter setting, in which the quantum Fisher information determines the precision limit. By virtue of the right (left) logarithmic derivative, one can construct an alternative quantum Fisher information matrix $J_\rmR(\theta)$ ($J_\rmL(\theta)$) and an alternative bound \cite{YuenL73,Hole82book} (see \aref{app:RLD}).

In the multiparameter setting, the bound in \eref{eq:QCR} usually  cannot be saturated due to information tradeoff among different parameters, which is a manifestation of the complementarity principle \cite{Bohr28}. A prominent  tradeoff relation is 
the Gill-Massar inequality \cite{GillM00},
\begin{align}\label{eq:GMI}
\tr [J_\rmS^{-1}(\theta)I(\theta,\scrA)]\leq d-1,
\end{align}
which is saturated iff $\scrA$ is rank 1. Besides quantum estimation theory, such tradeoff relations are useful to studying quantum incompatibility \cite{Zhu15IC,HeinJN22}, quantum steering \cite{ZhuHC16}, optimal quantum state estimation \cite{Zhu12the,LiFGK16, ZhuH18U,HouTSZ18}, and multiparameter quantum metrology \cite{SzczBD16,PezzCSH17,GessPS18,LuW21,GoldSF21,ChenCY22,AlbaD22}.

\emph{Fisher eigenvalues and concentration order.}---To decode quantum measurements from the Fisher information matrix, we need to extract the key properties that are  parametrization independent. A good starting point is the \emph{metric-adjusted Fisher information matrix} proposed by the author \cite{Zhu15IC} (called Fisher concord in \rcite{LiL16}),
\begin{align}\label{eq:FIMma}
I_\rmS(\theta,\scrA):=J_\rmS^{-1/2}(\theta)I(\theta,\scrA)J_\rmS^{-1/2}(\theta), 
\end{align}
note that $J_\rmS(\theta)$  defines a  metric in the state space \cite{BrauC94,Petz96,PetzS96}.
$I_\rmL(\theta,\scrA)$ and $I_\rmR(\theta,\scrA)$ can be defined in a similar way (see \aref{app:RLD}). Then \esref{eq:QCR}{eq:GMI} means
\begin{align}\label{eq:FIMmaQCRGMI}
I_\rmS(\theta,\scrA)\leq 1,\quad \tr I_\rmS(\theta,\scrA)\leq d-1.
\end{align}
The eigenvalues of $I_\rmS(\theta,\scrA)$ or $J_\rmS^{-1}(\theta)I(\theta,\scrA)$ equivalently are called the \emph{Fisher eigenvalues} of $\scrA$ at $\rho(\theta)$.
At a given quantum state, they are independent of the parametrization because $I(\theta,\scrA)$ and $J(\theta)$ experience a same congruence transformation when the  parametrization is changed.
In addition,  all Fisher eigenvalues lie in the closed interval $[0,1]$ thanks to \eref{eq:FIMmaQCRGMI}. The maximum Fisher eigenvalue 1 is of special interest because it means certain parameter can be estimated with the same precision as in the one-parameter setting. In other words, the metrology precision is not affected by other parameters, which is quite unusual given the constraint  of information complementarity.
The \emph{Fisher spectrum} of $\scrA$ at $\rho$, denoted by $\eig(\rho,\scrA)$, is the vector of Fisher eigenvalues with multiplicities taken into account. The \emph{Fisher purity} of $\scrA$ at $\rho$, denoted by $\wp(\rho,\scrA)$,  is defined as the sum of squared Fisher eigenvalues divided by $d-1$, that is, $ \wp(\rho,\scrA):=\tr(I_\rmS^2)/(d-1)$, which satisfies  $ \wp(\rho,\scrA)\leq 1$.

By virtue of the Fisher spectrum and  submajorization relation \cite{Bhat97book,MarsOA11book} (see \aref{app:Majorization}), we can introduce a concentration order on POVMs, which focuses on the degree of concentration in formation extraction and complements the  data-processing order. 
Let $\scrA, \scrB$ be two POVMs on $\caH$; then $\scrA$ is \emph{submajorized} by $\scrB$ at~$\rho$, denoted as  $\scrA\overset{\rho}{\preceq}_\rmw \scrB$, if $\eig(\rho,\scrA)$ is \emph{submajorized} by $\eig(\rho,\scrB)$, that is, $\eig(\rho,\scrA)\preceq_\rmw \eig(\rho,\scrB)$. 
Similarly,   $\scrA$ is majorized by $\scrB$, denoted as  $\scrA\overset{\rho}{\preceq} \scrB$, if $\eig(\rho,\scrA)$ is majorized by $\eig(\rho,\scrB)$. For both variants,   $\rho$ can be omitted if the order relation holds for all states in $\scrD(\caH)$.

\emph{Fisher sharp measurements.}---The POVM $\scrA$ is \emph{Fisher sharp} at $\rho(\theta)$ if $I_\rmS(\theta,\scrA)$ is a projector, that is,
all nonzero Fisher eigenvalues are equal to~1. It is \emph{universally Fisher sharp} if $I_\rmS(\theta,\scrA)$ is a projector whenever $\rho$ has full rank. Such POVMs are of special interest because they are maximum with respect to the concentration order thanks to \eref{eq:FIMmaQCRGMI}, which means the information extracted is as concentrated as possible. 
Incidentally, a Fisher-symmetric POVM is a POVM in the opposite extreme for which $I_\rmS(\theta,\scrA)$ is proportional to the identity matrix \cite{LiFGK16,ZhuH18U}. 
Such a POVM is interesting for different reasons: notably, a rank-1 Fisher-symmetric POVM
has the lowest Fisher purity among rank-1 POVMs and can extract the Fisher information as uniformly as possible. 
The \emph{sharpness index} of $\scrA$ at $\rho(\theta)$, denoted by $\zeta_\rho(\scrA)$, is  the multiplicity of the eigenvalue 1 in $I_\rmS(\theta,\scrA)$. 
Both 
Fisher sharp measurements and sharpness index are parametrization independent. Surprisingly, they are also independent of the state $\rho$ as long as $\rho$ is nonsingular.

\begin{figure}[t]
	\includegraphics[width=7cm]{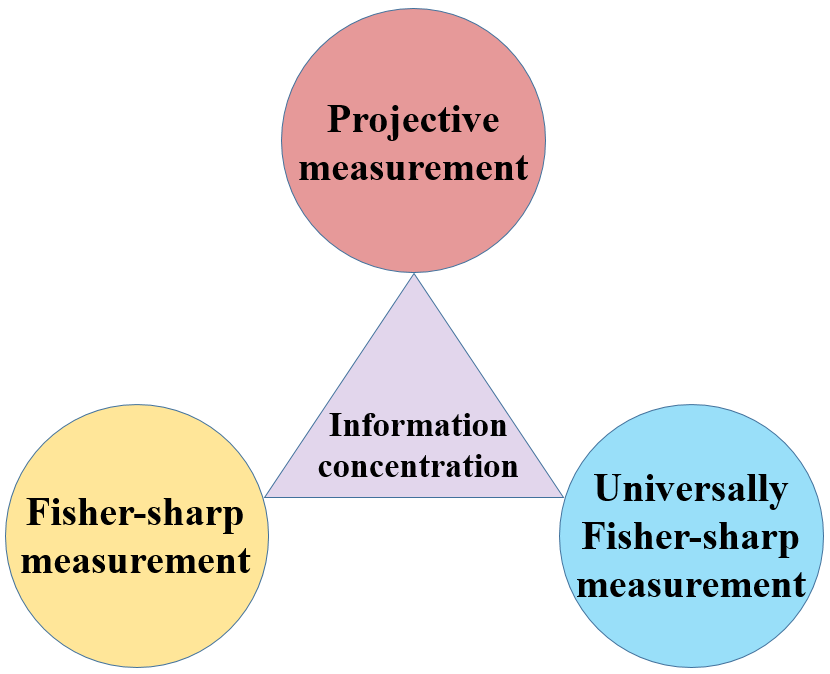}
	\caption{\label{fig:FisherSharp} 
		Equivalence of projective measurements, Fisher-sharp measurements, and universally  Fisher-sharp measurements based on the order of information concentration. 
	} 
\end{figure}

\begin{thm}\label{thm:SharpProjIndex}
	Suppose $\scrA$ is a simple POVM on $\caH$. Then 
	\begin{align}
	\zeta_\rho(\scrA)=\gamma(\scrA)-1\leq d-1 \quad \forall \rho\in \scrD(\caH);
	\end{align}	
the upper bound is saturated iff $\scrA$ is a rank-1 PVM.		
\end{thm}

\begin{thm}\label{thm:FisherSharp}
	Suppose $\scrA$ is a simple POVM on $\caH$. Then the following statements are equivalent. 
	\begin{enumerate}
		\item  $\scrA$ is a PVM.	
		\item  $\scrA$ is  Fisher sharp at some $\rho\in \scrD(\caH)$.	
		\item $\scrA$ is universally Fisher sharp.
	\end{enumerate}
\end{thm}
\Thref{thm:SharpProjIndex} establishes a precise connection between the sharpness index and projective index, which are seemingly disparate. \Thref{thm:FisherSharp} further reveals an intriguing connection between Fisher sharp measurements and projective measurements as illustrated in \fref{fig:FisherSharp}, which highlights the special role played by projective measurements in information extraction. 
Given \pref{pro:P(A)}, \thsref{thm:SharpProjIndex} and \ref{thm:FisherSharp}   are simplified versions of \thsref{thm:SharpProjIndex2} and \ref{thm:FisherSharp2}  below, which further show that the same conclusions still hold if the quantum Fisher information matrix in  \eref{eq:FIMma} is replaced by the analog based on the left or right logarithmic derivative. The following result is a simple corollary of \eref{eq:FIMmaQCRGMI} and \thref{thm:FisherSharp}  as well as \lref{lem:FisherRank} below.

\begin{corollary}\label{cor:FisherSharp2}
	Suppose  $\rho(\theta)\in \scrD(H)$,  $\scrA$ is a simple POVM on $\caH$, $r=\dim(\spa\scrA)$, and $p\geq 1$; then 
	\begin{align}
	\tr I_\rmS^p(\theta,\scrA)\leq r-1, \quad \tr I_\rmS^p(\theta,\scrA)\leq d-1.
	\end{align}
	The first inequality is saturated iff $\scrA$ is a PVM. When $p>1$, the second one is saturated iff $\scrA$ is a rank-1 PVM.
\end{corollary}
When $p=2$, \crref{cor:FisherSharp2} implies that
\begin{align}
\tr I_\rmS^2(\theta,\scrA)=\tr[J_\rmS^{-1}(\theta)I(\theta,\scrA)]^2\leq d-1,
\end{align}
and the inequality is saturated iff $\scrA$ is a rank-1 PVM,  in sharp contrast with the Gill-Massar inequality. This result offers a surprisingly  simple characterization of rank-1 projective measurements.
\begin{corollary}\label{cor:FisherPurity}
A simple POVM can attain the maximum Fisher purity 1 iff it is a rank-1 PVM.
\end{corollary}

\emph{Fisher eigenvalues and Fisher purity  for a qubit.}---As an illustration, suppose the qubit state $\rho$ is specified by the three components of the Bloch vector $\vec{s}=(x,y,z)$. Then the inverse quantum Fisher information matrix reads $J_\rmS^{-1}=1-\vec{s}\vec{s}$ \cite{Zhu15IC},
where $\vec{s}\vec{s}$ denotes the outer product of $\vec{s}$ and $\vec{s}$. Suppose the POVM has the form 
$\scrA=\{w_j(1+\vec{r}_j\cdot\vec{\sigma})\}_j$ with $w_j>0$, $\sum_j w_j=1$, and $\sum_j w_j \vec{r}_j=0$. Then the Fisher information matrix  reads
\begin{align}
I=\sum_j   \frac{w_j\vec{r}_j\vec{r}_j}{1+\vec{s}\cdot \vec{r}_j},
\end{align}
so the ranks of $I$ and $I_\rmS$ are both equal to the dimension spanned by $\vec{r}_j$ (cf. \lref{lem:FisherRank} below). 
Fisher eigenvalues, Gill-Massar trace, and Fisher purity can be computed with the following equations,
\begin{align}
J_\rmS^{-1}I &=\sum_{j} \frac{w_j[\vec{r}_j\vec{r}_j-(\vec{s}\cdot \vec{r}_j)\vec{s}\vec{r}_j]}{1+\vec{s}\cdot \vec{r}_j}, \\
\tr(J_\rmS^{-1}I)& =\sum_{j} \frac{w_j[\vec{r}_j\cdot\vec{r}_j-(\vec{s}\cdot \vec{r}_j)^2]}{1+\vec{s}\cdot \vec{r}_j}, \label{eq:GMTqubit}\\
\tr\bigl[(J_\rmS^{-1}I)^2\bigr] \!\!&=\sum_{j,k} \frac{w_jw_k[\vec{r}_j\cdot \vec{r}_k-(\vec{s}\cdot \vec{r}_j)(\vec{s}\cdot \vec{r}_k)]^2}
{(1+\vec{s}\cdot \vec{r}_j)(1+\vec{s}\cdot \vec{r}_k)}. \label{eq:FisherPurityQubit}
\end{align}
\Eref{eq:GMTqubit}  means $\tr(J_\rmS^{-1}I) \leq \sum_{j} w_j(1-\vec{s}\cdot \vec{r}_j)=1$,
and the inequality is saturated iff $\vec{r}_j=1$ for all $j$, so that $\scrA$ is rank-1, as expected from the Gill-Massar inequality.

\begin{figure}[b]
	\includegraphics[width=6cm]{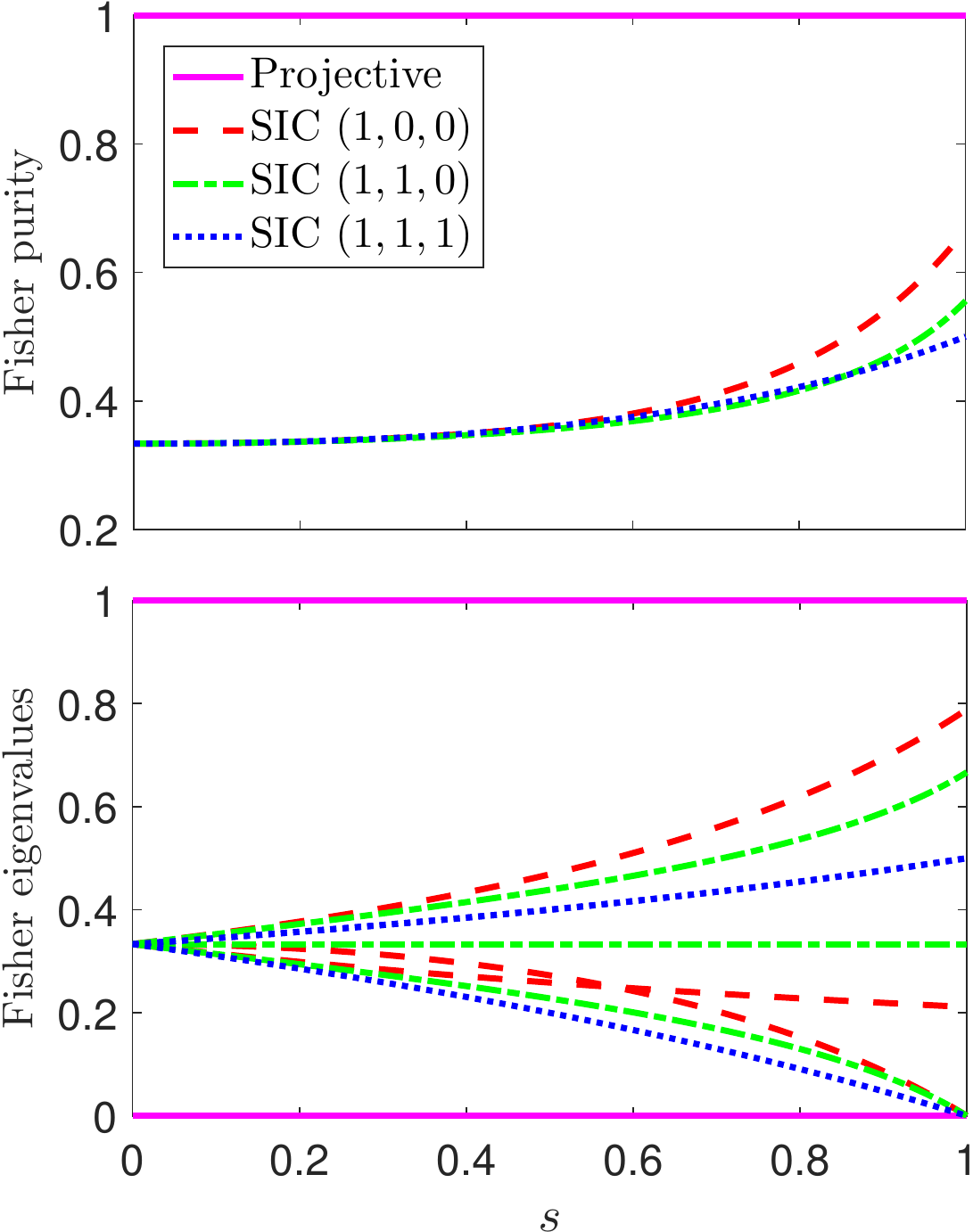}
	\caption{\label{fig:FisherEig}
		Fisher purities (upper plot) and Fisher eigenvalues (lower plot) of the projective measurement and SIC measurement for a qubit. The SIC is defined in \eref{eq:SIC}; the three Fisher eigenvalues depend on the length $s$ and direction (as indicated in the legend) of the Bloch vector $\vec{s}$  [in the direction $(1,1,1)$ two Fisher eigenvalues are identical due to symmetry]. 		
	} 
\end{figure}

If $\scrA$ is a rank-1 PVM, then  $I$ and $J_\rmS^{-1}I$ have rank~1, so $\scrA$ has only one nonzero Fisher eigenvalue, which is equal to 1, whenever $\rho$ is nonsingular. So  $\scrA$ is universally Fisher sharp. 
On the other hand, $\scrA$ is Fisher sharp at $\rho$ iff it has sharpness index 1 at $\rho$ (assuming $I\neq 0$). In that case, the Gill-Massar inequality is  saturated, so $\scrA$ is rank 1. In addition, $I$ is rank 1, so all $\vec{r}_j$ are parallel or antiparallel. Hence $\scrA$ is equivalent to a rank-1 PVM, which confirms \thsref{thm:SharpProjIndex} and \ref{thm:FisherSharp}.

Next, suppose $\scrA$ is a symmetric informationally complete POVM (SIC for short) \cite{Zaun11,ReneBSC04,FuchHS17}
 for comparison.  Let  $w_1=w_2=w_3=w_4=1/4$ and $\bm{r}_j$ for $j=1,2,3,4$ be the normalized Bloch vectors tied to the four directions
\begin{align}\label{eq:SIC}
(1,1,1),\;\; (1,-1,-1),\;\; (-1,1,-1),\;\; (-1,-1,1).
\end{align}
Then generically $\scrA$ has three distinct Fisher eigenvalues, which  depend on the length and direction of the Bloch vector $\vec{s}$. The Fisher purities and Fisher eigenvalues associated with the three  directions $(1,0,0)$, $(1,1,0)$, and $(1,1,1)$ are illustrated in \fref{fig:FisherEig} (see \aref{app:FisherEigQubit} for analytical formulas). Although the SIC is Fisher symmetric at the completely mixed state, it is in general not Fisher symmetric at other quantum states. Actually, it is impossible to construct a single-copy POVM that is universally Fisher symmetric.  By contrast, any projective measurement is universally Fisher sharp.

\emph{Quantum estimation theory based on superoperators.}---To prove our main results, we need
to reformulate quantum estimation theory in a way that is parametrization independent. This approach has found applications in quantum  estimation theory \cite{BrauC94, Zhu12the,Zhu14IOC,ZhuH18U},
geometry of quantum states \cite{Petz96,PetzS96}, and foundational studies on quantum incompatibility and steering \cite{Zhu15IC,ZhuHC16,HeinJN22}.
Here we further develop this approach and derive a number of  nontrivial results (see the Appendices), in addition to rederiving the old results, such as  \esref{eq:QCR}{eq:GMI}. Our study will be useful to multiparameter quantum metrology and various topics mentioned above. 
To simplify the discussion, henceforth we assume that $\{\rho_{,a}\}_{a=1}^g$ is an orthonormal basis for the space of traceless operators, given that Fisher eigenvalues are  parametrization independent.

Recall that the set of operators on $\caH$  forms a Hilbert space  with respect to the Hilbert-Schmidt inner product.   The vectors in  this space can be represented using the double-ket notation. For example, the operator $A$ can be expressed as $\dket{A}$, and the inner product between $A$ and $B$ can be expressed as $\dinner{A}{B}:=\tr(A^\dag B)$. Then  outer products of the form $\dket{A}\dbra{B}$ act as 
superoperators. Denote by $\mathbf{I}$ the identity superoperator and by $\bar{\mathbf{I}}$ the projector onto the space of traceless operators \cite{Zhu15IC,ZhuHC16}.

Given a nonsingular density operator $\rho\in \scrD(\caH)$ and POVM $\scrA=\{A_j\}_j$ on $\caH$, define
\begin{align}
\caF(\rho,\scrA)&:=\sum_{j|A_j\neq 0} \frac{\dket{A_j}\dbra{A_j}  }{\tr(\rho A_j)}, \label{eq:FDef} \\ \bar{\caF}(\rho,\scrA)&:=\bar{\mathbf{I}}\caF(\rho,\scrA)\bar{\mathbf{I}}= \sum_{j|A_j\neq 0} \frac{\dket{\bar{A}_j}\dbra{\bar{A}_j}  }{\tr(\rho A_j)}, \label{eq:FbarDef}
\end{align}
where $\bar{A}:=A-\tr(A)/d$ is the traceless part of $A$. Note that both $\caF(\rho,\scrA)$ and $\bar{\caF}(\rho,\scrA)$ are Hermitian and positive (semidefinite). Then the entries of the Fisher information matrix $I(\theta,\scrA)$ can be expressed as 
\begin{align}\label{eq:IFbar}
 I_{ab}(\theta, \scrA)=\dbra{\rho_{,a}}\caF(\rho,\scrA)\dket{\rho_{,b}}=\dbra{\rho_{,a}}\bar{\caF}(\rho,\scrA)\dket{\rho_{,b}}.
\end{align}
So  $I(\theta,\scrA)$ is the matrix representation of $\bar{\caF}(\rho,\scrA)$ in the operator basis  $\{\rho_{,a}\}_{a=1}^g$. By virtue of this connection, we can determine the ranks of $I(\theta,\scrA)$ and $I_\rmS(\theta,\scrA)$ and clarify the subtle relation between $I(\theta,\scrA)$ and $\caF(\rho,\scrA)$.

\begin{lem}\label{lem:FisherRank}
	Suppose $\rho(\theta)\in \scrD(H)$  and $\scrA$ is a POVM on $\caH$. Then 
	\begin{align}\label{eq:FisherRank}
	&\rk I_\rmS(\theta,\scrA)=	\rk I(\theta,\scrA)=\rk \bar{\caF}(\rho,\scrA)\nonumber\\
	&=\rk \caF(\rho,\scrA)-1=\dim \spa (\scrA)-1.
	\end{align}
\end{lem}

\begin{lem}\label{lem:FIequi}
	Suppose $\rho(\theta)\in \scrD(H)$  and $\scrA,\scrB$ are two POVMs on $\caH$. Then the  four equalities are equivalent. 
	\begin{align}\label{eq:FIequi}
	\begin{aligned}	
	\caF(\rho,\scrA)&= \caF(\rho,\scrB),&
	\quad \bar{\caF}(\rho,\scrA)&= \bar{\caF}(\rho,\scrB),\\	
	I(\theta,\scrA)&=I(\theta,\scrB), &\quad 
	I_\rmS(\theta,\scrA)&=I_\rmS(\theta,\scrB).
	\end{aligned}		
	\end{align}
\end{lem}

Next, we  introduce superoperator representations of quantum Fisher information matrices and metric-adjusted Fisher information matrices.  Define superoperators $\caL(\rho)$, $\caR(\rho)$, $\caS(\rho)$ as follows \cite{BrauC94,Petz96,PetzS96,Zhu12the},
\begin{align}
\caL(\rho)\dket{A}=\dket{\rho A}, \;\; 
\caR(\rho)\dket{A}=\dket{A \rho},\;\;
\caS=\frac{\caL+\caR}{2}. \label{eq:LRS}
\end{align}
Here $\rho$ can be omitted if there is no danger of confusion; $\caL,\caR,\caS$ are positive definite for $\rho\in \scrD(\caH)$. 
Define
\begin{align}
\caJ_\rmL&:=\caL^{-1}, &\quad \caJ_\rmR&:=\caR^{-1},& \quad \caJ_\rmS&:=\caS^{-1},  \label{eq:JLRS}\\
\bar{\caJ}_\rmL&:=\bid\caJ_\rmL\bid, &\quad \bar{\caJ}_\rmR&:=\bid\caJ_\rmR\bid,& \quad \bar{\caJ}_\rmS&:=\bid\caJ_\rmS\bid. 
\end{align}
Then 
\begin{align}
(J_{\rmS})_{ab}(\theta)=\dbra{\rho_{,a}}\caJ_\rmS(\rho)\dket{\rho_{,b}}=\dbra{\rho_{,a}}\bar{\caJ}_\rmS(\rho)\dket{\rho_{,b}}, 
\end{align}
so  $J_\rmS(\theta,\scrA)$ is the matrix representation of $\bar{\caJ}_\rmS(\rho,\scrA)$.

Finally, define
\begin{align}\label{eq:FFbarS}
\caF_\rmS:=\caJ_\rmS^{-1/2}\caF\caJ_\rmS^{-1/2},\quad
\bar{\caF}_\rmS&:=(\bar{\caJ}_\rmS^+)^{1/2}\bar{\caF}(\bar{\caJ}_\rmS^+)^{1/2},
\end{align}
where 
\begin{align}
\bar{\caJ}_\rmS^+=\caS-\douter{\rho}{\rho}=\caJ_\rmS^{-1}-\douter{\rho}{\rho}
\end{align}
is the Moore-Penrose generalized inverse \cite{Bern09book} of $\bar{\caJ}_\rmS$. Here the two equalities follow from \eref{eq:JLRS} and \lref{lem:JJbar} in \aref{app:FisherEig}. Then $I_\rmS(\theta,\scrA)$ is the matrix representation of $\bar{\caF}_\rmS(\rho,\scrA)$ and has the same nonzero  eigenvalues (including  multiplicities). Moreover, the spectrum of $\bar{\caF}_\rmS(\rho,\scrA)$ is  determined by the spectrum of $\caF_\rmS$ thanks to \lref{lem:FisherEigMul} in \aref{app:FisherEig}. 
These observations are crucial to proving our main results and are also useful to understanding quantum estimation theory in general. Superoperators  $\caF_\rmL, \caF_\rmR, \bar{\caF}_\rmL, \bar{\caF}_\rmR$ can be defined analogously.

\emph{Sharpness index and Fisher sharp measurements.}---By virtue of quantum estimation theory  presented above and  the Appendices we can now establish our main results. Given an operator or superoperator $\Theta$, 
denote by $\mu(\Theta)$ the algebraic multiplicity of the eigenvalue 1.

\begin{thm}\label{thm:SharpProjIndex2}
	Suppose $\rho(\theta)\in \scrD(\caH)$ and  $\scrA$ is a POVM on $\caH$. Then 
	\begin{gather}
	\mu(\caF_\rmS(\rho,\scrA))=\gamma(\scrA),  \label{eq:SharpProjIndex21}\\
	\zeta_\rho (\scrA)=\mu(I_\rmS(\theta,\scrA))
	=\mu(\bar{\caF}_\rmS(\rho,\scrA))=\gamma(\scrA)-1.\label{eq:SharpProjIndex22}
	\end{gather}	
	The same conclusions hold if $\caF_\rmS$ is replaced by $\caF_\rmL$ or $\caF_\rmR$, and $\bar{\caF}_\rmS$ ($I_\rmS$) is replaced by $\bar{\caF}_\rmL$ ($I_\rmR$) or $\bar{\caF}_\rmR$ ($I_\rmL$) accordingly. 
\end{thm}

\begin{thm}\label{thm:FisherSharp2}
	Suppose $\scrA$ is a simple POVM on $\caH$. Then the following five statements are equivalent. 
	\begin{enumerate}
		\item $\scrA$ is a PVM. 
		
		\item $\scrA$ is universally Fisher sharp. 
		
		\item $\scrA$ is Fisher sharp at some $\rho\in \scrD(\caH)$.

		\item $I_\rmS, I_\rmL, I_\rmR, \bar{\caF}_\rmS, \bar{\caF}_\rmL, \bar{\caF}_\rmR, \caF_\rmS, \caF_\rmL, \caF_\rmR$ are projectors at each $\rho\in \scrD(\caH)$. 
		
		\item One of the operators in statement 4 is a projector at some $\rho\in \scrD(\caH)$. 
	\end{enumerate}	
\end{thm}

\Thsref{thm:SharpProjIndex2} and \ref{thm:FisherSharp2} strengthen \thsref{thm:SharpProjIndex} and \ref{thm:FisherSharp}, respectively. They establish a precise connection between the sharpness index and projective index, which is independent of $\rho$ as long as $\rho$ is invertible. Consequently,  projective measurements,  Fisher-sharp measurements, and universally Fisher-sharp measurements are all equivalent to each other as illustrated in \fref{fig:FisherSharp}. Moreover, the same conclusions still hold if the quantum Fisher information matrix is built on the left or right logarithmic derivative instead of the SLD. The key for proving \Thsref{thm:SharpProjIndex2} and \ref{thm:FisherSharp2} is to transform statements about $I_\rmS$ into statements about the superoperator $\caF_\rmS$, which is parametrization independent and easier to analyze. Similar techniques are expected  to find applications in many other problems in quantum estimation theory and quantum metrology.

\emph{Summary.}---We introduced a simple order relation on quantum measurements based on the concentration of Fisher information, which complements the familiar data-processing order. Moreover, we showed that projective measurements happen to be Fisher-sharp measurements, which are extremal with respect to this concentration order. This result highlights the special role of projective measurements in information extraction, which has been overlooked for too long a time. In addition, we  introduced the concepts of Fisher eigenvalues and sharpness index, and showed that the  sharpness index  is completely determined by the projective index, that is, the number of outcomes in the finest projective measurement among the coarse graining of the given measurement. Our work endows projective measurements with a simple information theoretic interpretation, which is of intrinsic interest to foundational studies on 
quantum measurements, quantum  incompatibility, and information complementarity. Meanwhile, the ideas and technical tools we introduced are useful to studying
various problems in quantum estimation theory and multiparameter quantum metrology.

\bigskip

\acknowledgments
This work is  supported by   the National Natural Science Foundation of China (Grant No.~92165109),  National Key Research and Development Program of China (Grant No. 2022YFA1404204), and Shanghai Municipal Science and Technology Major Project (Grant No.~2019SHZDZX01).

\begin{appendix}

	\section{\label{app:Majorization}Majorization and submajorization}
	Let $x=(x_1,x_2, \ldots, x_n)$ and $y=(y_1,y_2, \ldots, y_n)$ be two real vectors in $\bbR^n$.  Recall that $x$ is submajorized by $y$, denoted as $x\preceq_\rmw y$ \cite{MarsOA11book,Bhat97book}, if 
	\begin{align}\label{eq:MajorizationDef}
	\sum_{j=1}^k x_j^{\downarrow}\leq \sum_{j=1}^k y_j^{\downarrow}\quad \forall k=1,2,\ldots, n,  
	\end{align}
	where $x_j^{\downarrow}$ ($y_j^{\downarrow}$) is the $j$th largest component of $x$ ($y$). 
	By contrast, $x$ is majorized by $y$, denoted as $x\preceq y$, if in addition the inequality in \eref{eq:MajorizationDef} is saturated in the case  $k=n$ \cite{MarsOA11book,Bhat97book}. On the other hand, $x$ and $y$ are majorization equivalent if $x\preceq y$ and $y\preceq x$ or equivalently if   $x\preceq_\rmw y$ and  $y\preceq_\rmw x$; in that case,  the components of $x$ and that of $y$ are identical up to a permutation.

	The submajorization relation can be generalized to Hermitian matrices, operators, and superoperators. Suppose  $A$ and $B$ are two Hermitian operators on a same Hilbert space for example. Then   $A$ is submajorized by $B$, also denoted as $A\preceq_\rmw B$, if $\eig(A)\preceq_\rmw \eig(B)$, where $\eig(A)$ ($\eig(B)$) is the 
	vector of eigenvalues of $A$ ($B$).  By definition $A\preceq_\rmw B$ whenever $A\leq B$.
	The majorization relation can be defined in a similar way.

	\section{\label{app:FisherEigQubit}Fisher  eigenvalues of the qubit SIC}
	Suppose $\scrA$ is the SIC defined by \eref{eq:SIC} and $\rho$ is determined by the Bloch vector $\vec{s}=(x,y,z)$. Then the Fisher purity follows from \eref{eq:FisherPurityQubit}, 
	\begin{align}
	\wp\!=\!1-\frac{2(9-7s^2+6\sqrt{3}\lsp xyz)}{3[(3-s^2)^2-4(x^2y^2+y^2z^2+z^2x^2)+8\sqrt{3}\lsp xyz]}. 
	\end{align} 
	In general, $\scrA$ has three distinct Fisher eigenvalues, which  depend on the length and direction of the Bloch vector~$\vec{s}$. In the three special directions $(1,0,0)$, $(1,1,0)$, and $(1,1,1)$, we can derive  analytical formulas (cf. \fref{fig:FisherEig}), 
	\begin{gather}
	\frac{1}{3-\sqrt{3}s},\quad \frac{1}{3+\sqrt{3}s}, \quad  \frac{1-s^2}{3-s^2};\\
	\frac{1}{3},\; \frac{3-2s^2+\sqrt{3s^2-2s^4}}{9-6s^2},\; \frac{3-2s^2-\sqrt{3s^2-2s^4}}{9-6s^2};\\
	\frac{1}{3-s},\quad \frac{1}{3-s},\quad \frac{1-s}{3-s}. 
	\end{gather}

	\section{\label{app:RLD}Metric adjusted Fisher information matrix based on the right logarithmic derivative}

	Suppose $\rho(\theta)\in \scrD(\caH)$  is specified by $g=d^2-1$ parameters $\theta_1, \theta_2, \ldots, \theta_g$ as in the main text. The right logarithmic derivative $L_{\rmR,a}$ of $\rho$ with respect to $\theta_a$ is determined by the equation $\rho_{,a}=\rho L_{\rmR,a}$ \cite{YuenL73,Hole82book}, 
	which means $L_{\rmR,a}=\rho^{-1}\rho_{,a}$. 
	The corresponding quantum Fisher information matrix $J_\rmR$  is a Hermitian matrix with  entries 
	\begin{align}
	(J_\rmR)_{ab}(\theta)&=\tr\bigl(\rho L_{\rmR,b}  L_{\rmR,a}^\dag\bigr)=\tr\bigl(\rho_{,a} \rho^{-1}\rho_{,b}\bigr)\nonumber\\
	&=\dbra{\rho_{,a}} \caL^{-1}(\rho) \dket{\rho_{,b}}=\dbra{\rho_{,a}} \caJ_\rmL(\rho) \dket{\rho_{,b}} \nonumber\\
	&=
	\dbra{\rho_{,a}} \bar{\caJ}_\rmL(\rho) \dket{\rho_{,b}}.
	\label{eq:QFIRLD}
	\end{align}
	The analog for the metric adjusted Fisher information matrix reads [cf. \eref{eq:FIMma}]
	\begin{align}\label{eq:FIMRma}
	I_\rmR(\theta,\scrA):=J_\rmR^{-1/2}(\theta)I(\theta,\scrA)J_\rmR^{-1/2}(\theta). 
	\end{align}
	Note that the spectrum  of $I_\rmR(\theta,\scrA)$  is independent of the parametrization, just as $I_\rmS(\theta,\scrA)$. The above equations also show that  $I_\rmR(\theta,\scrA)$ is the matrix representation of $\bar{\caF}_\rmL(\rho,\scrA)$ [not $\bar{\caF}_\rmR(\rho,\scrA)$] in the operator basis  $\{\rho_{,a}\}_{a=1}^g$, where
	\begin{align}
	\bar{\caF}_\rmL&:=(\bar{\caJ}_\rmL^+)^{1/2}\bar{\caF}(\bar{\caJ}_\rmL^+)^{1/2} \label{eq:FFbarL}
	\end{align}
	is defined in analogy to $\bar{\caF}_\rmS$ in 
	\eref{eq:FFbarS}.  So $I_\rmR(\theta,\scrA)$ and $\bar{\caF}_\rmL(\rho,\scrA)$ have the same nonzero spectrum.

	The left logarithmic derivative $L_{\rmL,a}$  is determined by the equation $\rho_{,a}= L_{\rmL,a} \rho$,
	which means $L_{\rmL,a}=\rho_{,a}\rho^{-1}$. 
	The corresponding quantum Fisher information matrix $J_\rmL$ is a Hermitian matrix with entries 
	\begin{align}
	(J_\rmL)_{ab}(\theta)&=\tr\bigl(\rho  L_{\rmL,a}^\dag L_{\rmL,b} \bigr)=\tr\bigl(\rho_{,a} \rho_{,b} \rho^{-1}\bigr)\nonumber\\
	&=\dbra{\rho_{,a}} \caR^{-1}(\rho) \dket{\rho_{,b}}=\dbra{\rho_{,a}} \caJ_\rmR(\rho) \dket{\rho_{,b}} \nonumber\\
	&=
	\dbra{\rho_{,a}} \bar{\caJ}_\rmR(\rho) \dket{\rho_{,b}}.  \label{eq:QFILLD}
	\end{align}
	The analog for the metric adjusted Fisher information matrix reads [cf. \esref{eq:FIMma}{eq:FIMRma}]
	\begin{align}
	I_\rmL(\theta,\scrA):=J_\rmL^{-1/2}(\theta)I(\theta,\scrA)J_\rmL^{-1/2}(\theta). 
	\end{align}
	Here $I_\rmL(\theta,\scrA)$ is the matrix representation of $\bar{\caF}_\rmR(\rho,\scrA)$  [not $\bar{\caF}_\rmL(\rho,\scrA)$] in the operator basis  $\{\rho_{,a}\}_{a=1}^g$, where
	\begin{align}
	\bar{\caF}_\rmR&:=(\bar{\caJ}_\rmR^+)^{1/2}\bar{\caF}(\bar{\caJ}_\rmR^+)^{1/2}.
	\end{align}
	So  $I_\rmL(\theta,\scrA)$ and $\bar{\caF}_\rmR(\rho,\scrA)$ have the same nonzero spectrum.

	\Esref{eq:QFIRLD}{eq:QFILLD} imply that $J_\rmL=J_\rmR^\rmT=J_\rmR^*$ and $I_\rmL=I_\rmR^\rmT=I_\rmR^*$. 
	So  $I_\rmL$ and $I_\rmR$ have the same spectrum;  $I_\rmL$ is a projector iff $I_\rmR$ is a projector. In addition, the inverse function is operator convex in the open interval $(0,\infty)$ \cite{Bhat97book}, so $\caS^{-1}\leq (\caL^{-1}+\caR^{-1})/2$,
	which means
	\begin{align}
	J_\rmS\leq  \frac{1}{2}(J_\rmR+J_\rmL)=\frac{1}{2}(J_\rmR+J_\rmR^*).
	\end{align}

	\section{\label{app:FisherAux}Auxiliary results on Fisher information matrices}
	
	Denote by $\Pi(\scrA)$ the projector onto the \emph{reconstruction subspace}, that is, $\spa(\scrA)$, and by $\rho(\scrA)$ the projection of $\rho$ into this space. Denote by $\caF^+(\rho,\scrA)$ the Moore-Penrose generalized inverse \cite{Bern09book} of $\caF(\rho,\scrA)$, which coincides with the usual inverse when $\caF(\rho,\scrA)$ is invertible.

	\begin{lem}\label{lem:FFbar}
		Suppose $\rho\in \scrD(\caH)$ and $\scrA$ is a POVM on $\caH$. Then 
		\begin{align}\label{eq:FFbar}
		\bar{\caF}^+(\rho,\scrA)=\caF^+(\rho,\scrA)-\douter{\rho(\scrA)}{\rho(\scrA)}.
		\end{align}
	\end{lem}
	
	\Lref{lem:FFbar} shows that $\caF(\rho,\scrA)$ is uniquely determined by $\bar{\caF}(\rho,\scrA)$, which implies \lref{lem:FIequi} in the main text.

	\begin{proposition}\label{pro:ProjEqui}
		Suppose $\rho(\theta)\in \scrD(H)$  and $\scrA$, $\scrB$ are two PVMs on $\caH$ that contain no zero operators. Then the following three statements are equivalent. 
		\begin{enumerate}
			\item $\scrB$ is a permutation of $\scrA$. 	
			\item $I(\theta,\scrA)=I(\theta,\scrB)$.
			\item $\caF(\rho,\scrA)=\caF(\rho, \scrB)$.
		\end{enumerate}
	\end{proposition}
	
	\Pref{pro:ProjEqui} shows that the map from  PVMs to Fisher information matrices is essentially injective, although this is not the case for general POVMs. In the special case $\rho=1/d$, the equivalence of statements 1 and 3 also follows from Proposition 5.5 in \rcite{HeinJN22}.

	\section{\label{app:FisherEig}Auxiliary results on Fisher eigenvalues}

	\begin{lem}\label{lem:JJbar}
		Suppose $\rho\in \scrD(\caH)$. Then 
		\begin{align}
		\bar{\caJ}_\rmS^+&=\caS-\douter{\rho}{\rho}, \; 
		\bar{\caJ}_\rmL^+=\caL-\douter{\rho}{\rho}, \;	
		\bar{\caJ}_\rmR^+=\caR-\douter{\rho}{\rho}.	\label{eq:JJbar}
		\end{align}	
	\end{lem}
	In conjunction with \eref{eq:JLRS}, this lemma clarifies the relation between $\bar{\caJ}_\rmS^+, \bar{\caJ}_\rmL^+, \bar{\caJ}_\rmR^+$ and $\caJ_\rmS^{-1}, \caJ_\rmL^{-1}, \caJ_\rmR^{-1}$, which is useful to studying Fisher eigenvalues.

	\begin{lem}\label{lem:GMgen}
		Suppose $\rho(\theta)\in \scrD(\caH)$ and $\scrA$ is a POVM on $\caH$ with $m$ POVM elements. Then 
		\begin{align}
		&\tr(I_\rmS(\theta,\scrA))+1=\tr(\bar{\caF}_\rmS(\scrA))+1=\tr(\bar{\caJ}_\rmS^+\caF(\scrA))+1\nonumber\\
		&=\tr(\caF_\rmS(\scrA))=\sum_{j, A_j\neq 0} \frac{\tr(\rho A_j^2 )}{\tr(\rho A_j)}\leq \min\{m,d\}. \label{eq:GMJF}
		\end{align}	
		When $m\leq d$, the inequality is saturated iff $\scrA$ is a PVM with $m$ nonzero POVM elements; when $d\leq m$, the inequality is saturated iff $\scrA$ is rank 1. The same conclusions hold if $I_\rmS$ is replaced by $I_\rmL$ or $I_\rmR$,
		and $\bar{\caF}_\rmS$, $\caF_\rmS$, $\bar{\caJ}_\rmS$ are replaced accordingly. 
	\end{lem}

	\Lref{lem:GMgen}  strengthens the Gill-Massar inequality in \eref{eq:GMI}, but our derivation is much simpler. The inequality $\tr(\caF_\rmS(\scrA))\leq m$ and saturation condition also follow from Proposition 5.1 in \rcite{HeinJN22}. However, the subtle connection between $\caF_\rmS(\scrA)$ and $I_\rmS(\theta,\scrA)$ was not  clarified before the current study.

	\begin{lem}\label{lem:FSLRmajor}
		Suppose $\rho\in \scrD(\caH)$ and  $\scrA$ is a POVM on a subspace $\caV\leq \caH$ with $\dim \caV\geq 1$. Then
		\begin{align}\label{eq:FSLRmajor}
		\caF_\rmS(\rho,\scrA)\preceq 	\caF_\rmL(\rho,\scrA)\simeq 	\caF_\rmR(\rho,\scrA),\\
		\bar{\caF}_\rmS(\rho,\scrA)\preceq 	\bar{\caF}_\rmL(\rho,\scrA)\simeq 	\bar{\caF}_\rmR(\rho,\scrA).	\label{eq:FSLRbarMajor}	
		\end{align}	
	\end{lem}
	\Eref{eq:FSLRmajor} means $\caF_\rmS$ is majorized by $\caF_\rmL$, which is similar to $\caF_\rmR$.

	\begin{lem}\label{lem:FLRnorm}
		Suppose $\rho\in \scrD(\caH)$ and  $\scrA$ is a POVM on $\caV\leq \caH$ with $\dim \caV\geq 1$. Then
		\begin{gather} 
		\caF(\rho,\scrA)\leq \caJ_\rmS(\rho),\,  \caJ_\rmL(\rho), \, \caJ_\rmR(\rho),  \label{eq:QCRF} \\
		\|\caF_\rmS(\rho,\scrA)\|= \|\caF_\rmL(\rho,\scrA)\|=\|\caF_\rmR(\rho,\scrA)\|=1. \label{eq:FLRnorm}
		\end{gather}	
		If $\scrA$ is irreducible, then the eigenvalue 1 is nondegenerate for   $\caF_\rmS(\rho,\scrA)$, $\caF_\rmL(\rho,\scrA)$,  and $\caF_\rmR(\rho,\scrA)$. 
	\end{lem}
	\Eref{eq:FLRnorm} implies \eref{eq:QCRF}, which in turn implies the SLD bound in \eref{eq:QCR}.

	Given an operator or superoperator $\Theta$ and  a real number $\lambda$,
	denote by $\mu_\lambda(\Theta)$ the algebraic multiplicity of the eigenvalue $\lambda$ [by convention $\mu_\lambda(\Theta)=0$ if $\lambda$ is not an eigenvalue]; then  $\mu_1(\Theta)=\mu(\Theta)$. The following lemma shows that the Fisher spectrum of $\scrA$ is  determined by the spectrum of $\caF_\rmS$, which is usually easier to analyze. 
	\begin{lem}\label{lem:FisherEigMul}
		Suppose $\rho\in \scrD(\caH)$ and  $\scrA$ is a POVM on $\caH$. Then 
		\begin{align}\label{eq:FisherEigMul}
		\mu_\lambda(\bar{\caF}_\rmS(\rho,\scrA))=\begin{cases}
		\mu_\lambda(\caF_\rmS(\rho,\scrA))+1 &\lambda=0,
		\\\mu_\lambda(\caF_\rmS(\rho,\scrA))-1 &\lambda=1,\\
		\mu_\lambda(\caF_\rmS(\rho,\scrA)) &\lambda\neq 0,1; \\
		\end{cases}
		\end{align}	
		the same conclusions hold if  $\bar{\caF}_\rmS$ is replaced by $\bar{\caF}_\rmL$ or $\bar{\caF}_\rmR$, and
		$\caF_\rmS$ is replaced by $\caF_\rmL$ or $\caF_\rmR$  accordingly. 
	\end{lem}
	
\end{appendix}

\bibliography{all_references}

\begin{thebibliography}{40}%
\makeatletter
\providecommand \@ifxundefined [1]{%
 \@ifx{#1\undefined}
}%
\providecommand \@ifnum [1]{%
 \ifnum #1\expandafter \@firstoftwo
 \else \expandafter \@secondoftwo
 \fi
}%
\providecommand \@ifx [1]{%
 \ifx #1\expandafter \@firstoftwo
 \else \expandafter \@secondoftwo
 \fi
}%
\providecommand \natexlab [1]{#1}%
\providecommand \enquote  [1]{``#1''}%
\providecommand \bibnamefont  [1]{#1}%
\providecommand \bibfnamefont [1]{#1}%
\providecommand \citenamefont [1]{#1}%
\providecommand \href@noop [0]{\@secondoftwo}%
\providecommand \href [0]{\begingroup \@sanitize@url \@href}%
\providecommand \@href[1]{\@@startlink{#1}\@@href}%
\providecommand \@@href[1]{\endgroup#1\@@endlink}%
\providecommand \@sanitize@url [0]{\catcode `\\12\catcode `\$12\catcode
  `\&12\catcode `\#12\catcode `\^12\catcode `\_12\catcode `\%12\relax}%
\providecommand \@@startlink[1]{}%
\providecommand \@@endlink[0]{}%
\providecommand \url  [0]{\begingroup\@sanitize@url \@url }%
\providecommand \@url [1]{\endgroup\@href {#1}{\urlprefix }}%
\providecommand \urlprefix  [0]{URL }%
\providecommand \Eprint [0]{\href }%
\providecommand \doibase [0]{https://doi.org/}%
\providecommand \selectlanguage [0]{\@gobble}%
\providecommand \bibinfo  [0]{\@secondoftwo}%
\providecommand \bibfield  [0]{\@secondoftwo}%
\providecommand \translation [1]{[#1]}%
\providecommand \BibitemOpen [0]{}%
\providecommand \bibitemStop [0]{}%
\providecommand \bibitemNoStop [0]{.\EOS\space}%
\providecommand \EOS [0]{\spacefactor3000\relax}%
\providecommand \BibitemShut  [1]{\csname bibitem#1\endcsname}%
\let\auto@bib@innerbib\@empty
\bibitem [{\citenamefont {{von Neumann}}(1955)}]{Neum55}%
  \BibitemOpen
  \bibfield  {author} {\bibinfo {author} {\bibfnamefont {J.}~\bibnamefont {{von
  Neumann}}},\ }\href@noop {} {\emph {\bibinfo {title} {Mathematical
  Foundations of Quantum Mechanics}}}\ (\bibinfo  {publisher} {Princeton
  University Press},\ \bibinfo {address} {Princeton, NJ},\ \bibinfo {year}
  {1955})\ \bibinfo {note} {translated from the German edition by R. T.
  Beyer}\BibitemShut {NoStop}%
\bibitem [{\citenamefont {Nielsen}\ and\ \citenamefont
  {Chuang}(2010)}]{NielC10book}%
  \BibitemOpen
  \bibfield  {author} {\bibinfo {author} {\bibfnamefont {M.~A.}\ \bibnamefont
  {Nielsen}}\ and\ \bibinfo {author} {\bibfnamefont {I.~L.}\ \bibnamefont
  {Chuang}},\ }\href@noop {} {\emph {\bibinfo {title} {Quantum Computation and
  Quantum Information}}}\ (\bibinfo  {publisher} {Cambridge University Press},\
  \bibinfo {address} {Cambridge, UK},\ \bibinfo {year} {2010})\BibitemShut
  {NoStop}%
\bibitem [{\citenamefont {Busch}\ \emph {et~al.}(2016)\citenamefont {Busch},
  \citenamefont {Lahti}, \citenamefont {Pellonp\"a\"a},\ and\ \citenamefont
  {Ylinen}}]{BuscLPY16}%
  \BibitemOpen
  \bibfield  {author} {\bibinfo {author} {\bibfnamefont {P.}~\bibnamefont
  {Busch}}, \bibinfo {author} {\bibfnamefont {P.}~\bibnamefont {Lahti}},
  \bibinfo {author} {\bibfnamefont {J.-P.}\ \bibnamefont {Pellonp\"a\"a}},\
  and\ \bibinfo {author} {\bibfnamefont {K.}~\bibnamefont {Ylinen}},\
  }\href@noop {} {\emph {\bibinfo {title} {Quantum Measurement}}}\ (\bibinfo
  {publisher} {Springer},\ \bibinfo {address} {Switzerland},\ \bibinfo {year}
  {2016})\BibitemShut {NoStop}%
\bibitem [{\citenamefont {Martens}\ and\ \citenamefont
  {de~Muynck}(1990)}]{MartM90}%
  \BibitemOpen
  \bibfield  {author} {\bibinfo {author} {\bibfnamefont {H.}~\bibnamefont
  {Martens}}\ and\ \bibinfo {author} {\bibfnamefont {W.~M.}\ \bibnamefont
  {de~Muynck}},\ }\bibfield  {title} {\bibinfo {title} {Nonideal quantum
  measurements},\ }\href@noop {} {\bibfield  {journal} {\bibinfo  {journal}
  {Found. Phys.}\ }\textbf {\bibinfo {volume} {20}},\ \bibinfo {pages} {255}
  (\bibinfo {year} {1990})}\BibitemShut {NoStop}%
\bibitem [{\citenamefont {Zhu}(2015)}]{Zhu15IC}%
  \BibitemOpen
  \bibfield  {author} {\bibinfo {author} {\bibfnamefont {H.}~\bibnamefont
  {Zhu}},\ }\bibfield  {title} {\bibinfo {title} {{Information complementarity:
  A new paradigm for decoding quantum incompatibility}},\ }\href@noop {}
  {\bibfield  {journal} {\bibinfo  {journal} {Sci. Rep.}\ }\textbf {\bibinfo
  {volume} {5}},\ \bibinfo {pages} {14317} (\bibinfo {year}
  {2015})}\BibitemShut {NoStop}%
\bibitem [{\citenamefont {Zhu}\ \emph {et~al.}(2016)\citenamefont {Zhu},
  \citenamefont {Hayashi},\ and\ \citenamefont {Chen}}]{ZhuHC16}%
  \BibitemOpen
  \bibfield  {author} {\bibinfo {author} {\bibfnamefont {H.}~\bibnamefont
  {Zhu}}, \bibinfo {author} {\bibfnamefont {M.}~\bibnamefont {Hayashi}},\ and\
  \bibinfo {author} {\bibfnamefont {L.}~\bibnamefont {Chen}},\ }\bibfield
  {title} {\bibinfo {title} {Universal steering criteria},\ }\href@noop {}
  {\bibfield  {journal} {\bibinfo  {journal} {Phys. Rev. Lett.}\ }\textbf
  {\bibinfo {volume} {116}},\ \bibinfo {pages} {070403} (\bibinfo {year}
  {2016})}\BibitemShut {NoStop}%
\bibitem [{\citenamefont {Zhu}(2022)}]{Zhu22}%
  \BibitemOpen
  \bibfield  {author} {\bibinfo {author} {\bibfnamefont {H.}~\bibnamefont
  {Zhu}},\ }\bibfield  {title} {\bibinfo {title} {Quantum measurements in the
  light of quantum state estimation},\ }\href@noop {} {\bibfield  {journal}
  {\bibinfo  {journal} {PRX Quantum}\ }\textbf {\bibinfo {volume} {3}},\
  \bibinfo {pages} {030306} (\bibinfo {year} {2022})}\BibitemShut {NoStop}%
\bibitem [{\citenamefont {Heinosaari}\ \emph {et~al.}(2022)\citenamefont
  {Heinosaari}, \citenamefont {Jivulescu},\ and\ \citenamefont
  {Nechita}}]{HeinJN22}%
  \BibitemOpen
  \bibfield  {author} {\bibinfo {author} {\bibfnamefont {T.}~\bibnamefont
  {Heinosaari}}, \bibinfo {author} {\bibfnamefont {M.~A.}\ \bibnamefont
  {Jivulescu}},\ and\ \bibinfo {author} {\bibfnamefont {I.}~\bibnamefont
  {Nechita}},\ }\bibfield  {title} {\bibinfo {title} {Order preserving maps on
  quantum measurements},\ }\href@noop {} {\bibfield  {journal} {\bibinfo
  {journal} {Quantum}\ }\textbf {\bibinfo {volume} {6}},\ \bibinfo {pages}
  {853} (\bibinfo {year} {2022})}\BibitemShut {NoStop}%
\bibitem [{\citenamefont {Kuramochi}(2015)}]{Kura15}%
  \BibitemOpen
  \bibfield  {author} {\bibinfo {author} {\bibfnamefont {Y.}~\bibnamefont
  {Kuramochi}},\ }\bibfield  {title} {\bibinfo {title} {Minimal sufficient
  positive-operator valued measure on a separable {Hilbert} space},\
  }\href@noop {} {\bibfield  {journal} {\bibinfo  {journal} {J. Math. Phys.}\
  }\textbf {\bibinfo {volume} {56}},\ \bibinfo {pages} {102205} (\bibinfo
  {year} {2015})}\BibitemShut {NoStop}%
\bibitem [{\citenamefont {Fisher}(1925)}]{Fish25}%
  \BibitemOpen
  \bibfield  {author} {\bibinfo {author} {\bibfnamefont {R.~A.}\ \bibnamefont
  {Fisher}},\ }\bibfield  {title} {\bibinfo {title} {Theory of statistical
  estimation},\ }\href@noop {} {\bibfield  {journal} {\bibinfo  {journal}
  {Math. Proc. Cambr. Philos. Soc.}\ }\textbf {\bibinfo {volume} {22}},\
  \bibinfo {pages} {700} (\bibinfo {year} {1925})}\BibitemShut {NoStop}%
\bibitem [{\citenamefont {Rao}(2002)}]{Rao02book}%
  \BibitemOpen
  \bibfield  {author} {\bibinfo {author} {\bibfnamefont {C.~R.}\ \bibnamefont
  {Rao}},\ }\href@noop {} {\emph {\bibinfo {title} {Linear Statistical
  Inference and its Applications}}},\ Wiley Series in Probability and
  Statistics\ (\bibinfo  {publisher} {Wiley-Interscience},\ \bibinfo {year}
  {2002})\BibitemShut {NoStop}%
\bibitem [{\citenamefont {Helstrom}(1976)}]{Hels76book}%
  \BibitemOpen
  \bibfield  {author} {\bibinfo {author} {\bibfnamefont {C.~W.}\ \bibnamefont
  {Helstrom}},\ }\href@noop {} {\emph {\bibinfo {title} {Quantum Detection and
  Estimation Theory}}}\ (\bibinfo  {publisher} {Academic Press},\ \bibinfo
  {address} {New York},\ \bibinfo {year} {1976})\BibitemShut {NoStop}%
\bibitem [{\citenamefont {Holevo}(1982)}]{Hole82book}%
  \BibitemOpen
  \bibfield  {author} {\bibinfo {author} {\bibfnamefont {A.~S.}\ \bibnamefont
  {Holevo}},\ }\href@noop {} {\emph {\bibinfo {title} {Probabilistic and
  Statistical Aspects of Quantum Theory}}}\ (\bibinfo  {publisher}
  {North-Holland},\ \bibinfo {address} {Amsterdam},\ \bibinfo {year}
  {1982})\BibitemShut {NoStop}%
\bibitem [{\citenamefont {Braunstein}\ and\ \citenamefont
  {Caves}(1994)}]{BrauC94}%
  \BibitemOpen
  \bibfield  {author} {\bibinfo {author} {\bibfnamefont {S.~L.}\ \bibnamefont
  {Braunstein}}\ and\ \bibinfo {author} {\bibfnamefont {C.~M.}\ \bibnamefont
  {Caves}},\ }\bibfield  {title} {\bibinfo {title} {Statistical distance and
  the geometry of quantum states},\ }\href@noop {} {\bibfield  {journal}
  {\bibinfo  {journal} {Phys. Rev. Lett.}\ }\textbf {\bibinfo {volume} {72}},\
  \bibinfo {pages} {3439} (\bibinfo {year} {1994})}\BibitemShut {NoStop}%
\bibitem [{\citenamefont {Hayashi}(2005)}]{Haya05book}%
  \BibitemOpen
  \bibinfo {editor} {\bibfnamefont {M.}~\bibnamefont {Hayashi}},\ ed.,\
  \href@noop {} {\emph {\bibinfo {title} {Asymptotic Theory of Quantum
  Statistical Inference}}}\ (\bibinfo  {publisher} {World Scientific},\
  \bibinfo {address} {Singapore},\ \bibinfo {year} {2005})\BibitemShut
  {NoStop}%
\bibitem [{\citenamefont {Liu}\ \emph {et~al.}(2019)\citenamefont {Liu},
  \citenamefont {Yuan}, \citenamefont {Lu},\ and\ \citenamefont
  {Wang}}]{LiuYLW19}%
  \BibitemOpen
  \bibfield  {author} {\bibinfo {author} {\bibfnamefont {J.}~\bibnamefont
  {Liu}}, \bibinfo {author} {\bibfnamefont {H.}~\bibnamefont {Yuan}}, \bibinfo
  {author} {\bibfnamefont {X.-M.}\ \bibnamefont {Lu}},\ and\ \bibinfo {author}
  {\bibfnamefont {X.}~\bibnamefont {Wang}},\ }\bibfield  {title} {\bibinfo
  {title} {Quantum {Fisher} information matrix and multiparameter estimation},\
  }\href@noop {} {\bibfield  {journal} {\bibinfo  {journal} {J. Phys. A: Math.
  Theor.}\ }\textbf {\bibinfo {volume} {53}},\ \bibinfo {pages} {023001}
  (\bibinfo {year} {2019})}\BibitemShut {NoStop}%
\bibitem [{\citenamefont {Yuen}\ and\ \citenamefont {Lax}(1973)}]{YuenL73}%
  \BibitemOpen
  \bibfield  {author} {\bibinfo {author} {\bibfnamefont {H.~P.}\ \bibnamefont
  {Yuen}}\ and\ \bibinfo {author} {\bibfnamefont {M.}~\bibnamefont {Lax}},\
  }\bibfield  {title} {\bibinfo {title} {Multiple-parameter quantum estimation
  and measurement of nonselfadjoint observables},\ }\href@noop {} {\bibfield
  {journal} {\bibinfo  {journal} {IEEE Trans. Inf. Theory}\ }\textbf {\bibinfo
  {volume} {19}},\ \bibinfo {pages} {740} (\bibinfo {year} {1973})}\BibitemShut
  {NoStop}%
\bibitem [{\citenamefont {Bohr}(1928)}]{Bohr28}%
  \BibitemOpen
  \bibfield  {author} {\bibinfo {author} {\bibfnamefont {N.}~\bibnamefont
  {Bohr}},\ }\bibfield  {title} {\bibinfo {title} {The quantum postulate and
  the recent development of atomic theory},\ }\href@noop {} {\bibfield
  {journal} {\bibinfo  {journal} {Nature}\ }\textbf {\bibinfo {volume} {121}},\
  \bibinfo {pages} {580} (\bibinfo {year} {1928})}\BibitemShut {NoStop}%
\bibitem [{\citenamefont {Gill}\ and\ \citenamefont {Massar}(2000)}]{GillM00}%
  \BibitemOpen
  \bibfield  {author} {\bibinfo {author} {\bibfnamefont {R.~D.}\ \bibnamefont
  {Gill}}\ and\ \bibinfo {author} {\bibfnamefont {S.}~\bibnamefont {Massar}},\
  }\bibfield  {title} {\bibinfo {title} {State estimation for large
  ensembles},\ }\href@noop {} {\bibfield  {journal} {\bibinfo  {journal} {Phys.
  Rev. A}\ }\textbf {\bibinfo {volume} {61}},\ \bibinfo {pages} {042312}
  (\bibinfo {year} {2000})}\BibitemShut {NoStop}%
\bibitem [{\citenamefont {Zhu}(2012)}]{Zhu12the}%
  \BibitemOpen
  \bibfield  {author} {\bibinfo {author} {\bibfnamefont {H.}~\bibnamefont
  {Zhu}},\ }\emph {\bibinfo {title} {Quantum State Estimation and Symmetric
  Informationally Complete {POM}s}},\ \href
  {http://scholarbank.nus.edu.sg/bitstream/10635/35247/1/ZhuHJthesis.pdf}
  {Ph.D. thesis},\ \bibinfo  {school} {National University of Singapore}
  (\bibinfo {year} {2012})\BibitemShut {NoStop}%
\bibitem [{\citenamefont {Li}\ \emph {et~al.}(2016)\citenamefont {Li},
  \citenamefont {Ferrie}, \citenamefont {Gross}, \citenamefont {Kalev},\ and\
  \citenamefont {Caves}}]{LiFGK16}%
  \BibitemOpen
  \bibfield  {author} {\bibinfo {author} {\bibfnamefont {N.}~\bibnamefont
  {Li}}, \bibinfo {author} {\bibfnamefont {C.}~\bibnamefont {Ferrie}}, \bibinfo
  {author} {\bibfnamefont {J.~A.}\ \bibnamefont {Gross}}, \bibinfo {author}
  {\bibfnamefont {A.}~\bibnamefont {Kalev}},\ and\ \bibinfo {author}
  {\bibfnamefont {C.~M.}\ \bibnamefont {Caves}},\ }\bibfield  {title} {\bibinfo
  {title} {Fisher-symmetric informationally complete measurements for pure
  states},\ }\href@noop {} {\bibfield  {journal} {\bibinfo  {journal} {Phys.
  Rev. Lett.}\ }\textbf {\bibinfo {volume} {116}},\ \bibinfo {pages} {180402}
  (\bibinfo {year} {2016})}\BibitemShut {NoStop}%
\bibitem [{\citenamefont {Zhu}\ and\ \citenamefont {Hayashi}(2018)}]{ZhuH18U}%
  \BibitemOpen
  \bibfield  {author} {\bibinfo {author} {\bibfnamefont {H.}~\bibnamefont
  {Zhu}}\ and\ \bibinfo {author} {\bibfnamefont {M.}~\bibnamefont {Hayashi}},\
  }\bibfield  {title} {\bibinfo {title} {Universally {Fisher}-symmetric
  informationally complete measurements},\ }\href@noop {} {\bibfield  {journal}
  {\bibinfo  {journal} {Phys. Rev. Lett.}\ }\textbf {\bibinfo {volume} {120}},\
  \bibinfo {pages} {030404} (\bibinfo {year} {2018})}\BibitemShut {NoStop}%
\bibitem [{\citenamefont {Hou}\ \emph {et~al.}(2018)\citenamefont {Hou},
  \citenamefont {Tang}, \citenamefont {Shang}, \citenamefont {Zhu},
  \citenamefont {Li}, \citenamefont {Yuan}, \citenamefont {Wu}, \citenamefont
  {Xiang}, \citenamefont {Li},\ and\ \citenamefont {Guo}}]{HouTSZ18}%
  \BibitemOpen
  \bibfield  {author} {\bibinfo {author} {\bibfnamefont {Z.}~\bibnamefont
  {Hou}}, \bibinfo {author} {\bibfnamefont {J.-F.}\ \bibnamefont {Tang}},
  \bibinfo {author} {\bibfnamefont {J.}~\bibnamefont {Shang}}, \bibinfo
  {author} {\bibfnamefont {H.}~\bibnamefont {Zhu}}, \bibinfo {author}
  {\bibfnamefont {J.}~\bibnamefont {Li}}, \bibinfo {author} {\bibfnamefont
  {Y.}~\bibnamefont {Yuan}}, \bibinfo {author} {\bibfnamefont {K.-D.}\
  \bibnamefont {Wu}}, \bibinfo {author} {\bibfnamefont {G.-Y.}\ \bibnamefont
  {Xiang}}, \bibinfo {author} {\bibfnamefont {C.-F.}\ \bibnamefont {Li}},\ and\
  \bibinfo {author} {\bibfnamefont {G.-C.}\ \bibnamefont {Guo}},\ }\bibfield
  {title} {\bibinfo {title} {Deterministic realization of collective
  measurements via photonic quantum walks},\ }\href@noop {} {\bibfield
  {journal} {\bibinfo  {journal} {Nat. Commun.}\ }\textbf {\bibinfo {volume}
  {9}},\ \bibinfo {pages} {1414} (\bibinfo {year} {2018})}\BibitemShut
  {NoStop}%
\bibitem [{\citenamefont {Szczykulska}\ \emph {et~al.}(2016)\citenamefont
  {Szczykulska}, \citenamefont {Baumgratz},\ and\ \citenamefont
  {Datta}}]{SzczBD16}%
  \BibitemOpen
  \bibfield  {author} {\bibinfo {author} {\bibfnamefont {M.}~\bibnamefont
  {Szczykulska}}, \bibinfo {author} {\bibfnamefont {T.}~\bibnamefont
  {Baumgratz}},\ and\ \bibinfo {author} {\bibfnamefont {A.}~\bibnamefont
  {Datta}},\ }\bibfield  {title} {\bibinfo {title} {Multi-parameter quantum
  metrology},\ }\href@noop {} {\bibfield  {journal} {\bibinfo  {journal}
  {Advances in Physics: X}\ }\textbf {\bibinfo {volume} {1}},\ \bibinfo {pages}
  {621} (\bibinfo {year} {2016})}\BibitemShut {NoStop}%
\bibitem [{\citenamefont {Pezz\`e}\ \emph {et~al.}(2017)\citenamefont
  {Pezz\`e}, \citenamefont {Ciampini}, \citenamefont {Spagnolo}, \citenamefont
  {Humphreys}, \citenamefont {Datta}, \citenamefont {Walmsley}, \citenamefont
  {Barbieri}, \citenamefont {Sciarrino},\ and\ \citenamefont
  {Smerzi}}]{PezzCSH17}%
  \BibitemOpen
  \bibfield  {author} {\bibinfo {author} {\bibfnamefont {L.}~\bibnamefont
  {Pezz\`e}}, \bibinfo {author} {\bibfnamefont {M.~A.}\ \bibnamefont
  {Ciampini}}, \bibinfo {author} {\bibfnamefont {N.}~\bibnamefont {Spagnolo}},
  \bibinfo {author} {\bibfnamefont {P.~C.}\ \bibnamefont {Humphreys}}, \bibinfo
  {author} {\bibfnamefont {A.}~\bibnamefont {Datta}}, \bibinfo {author}
  {\bibfnamefont {I.~A.}\ \bibnamefont {Walmsley}}, \bibinfo {author}
  {\bibfnamefont {M.}~\bibnamefont {Barbieri}}, \bibinfo {author}
  {\bibfnamefont {F.}~\bibnamefont {Sciarrino}},\ and\ \bibinfo {author}
  {\bibfnamefont {A.}~\bibnamefont {Smerzi}},\ }\bibfield  {title} {\bibinfo
  {title} {Optimal measurements for simultaneous quantum estimation of multiple
  phases},\ }\href@noop {} {\bibfield  {journal} {\bibinfo  {journal} {Phys.
  Rev. Lett.}\ }\textbf {\bibinfo {volume} {119}},\ \bibinfo {pages} {130504}
  (\bibinfo {year} {2017})}\BibitemShut {NoStop}%
\bibitem [{\citenamefont {Gessner}\ \emph {et~al.}(2018)\citenamefont
  {Gessner}, \citenamefont {Pezz\`e},\ and\ \citenamefont {Smerzi}}]{GessPS18}%
  \BibitemOpen
  \bibfield  {author} {\bibinfo {author} {\bibfnamefont {M.}~\bibnamefont
  {Gessner}}, \bibinfo {author} {\bibfnamefont {L.}~\bibnamefont {Pezz\`e}},\
  and\ \bibinfo {author} {\bibfnamefont {A.}~\bibnamefont {Smerzi}},\
  }\bibfield  {title} {\bibinfo {title} {Sensitivity bounds for multiparameter
  quantum metrology},\ }\href@noop {} {\bibfield  {journal} {\bibinfo
  {journal} {Phys. Rev. Lett.}\ }\textbf {\bibinfo {volume} {121}},\ \bibinfo
  {pages} {130503} (\bibinfo {year} {2018})}\BibitemShut {NoStop}%
\bibitem [{\citenamefont {Lu}\ and\ \citenamefont {Wang}(2021)}]{LuW21}%
  \BibitemOpen
  \bibfield  {author} {\bibinfo {author} {\bibfnamefont {X.-M.}\ \bibnamefont
  {Lu}}\ and\ \bibinfo {author} {\bibfnamefont {X.}~\bibnamefont {Wang}},\
  }\bibfield  {title} {\bibinfo {title} {Incorporating {Heisenberg's}
  uncertainty principle into quantum multiparameter estimation},\ }\href@noop
  {} {\bibfield  {journal} {\bibinfo  {journal} {Phys. Rev. Lett.}\ }\textbf
  {\bibinfo {volume} {126}},\ \bibinfo {pages} {120503} (\bibinfo {year}
  {2021})}\BibitemShut {NoStop}%
\bibitem [{\citenamefont {Goldberg}\ \emph {et~al.}(2021)\citenamefont
  {Goldberg}, \citenamefont {S\'anchez-Soto},\ and\ \citenamefont
  {Ferretti}}]{GoldSF21}%
  \BibitemOpen
  \bibfield  {author} {\bibinfo {author} {\bibfnamefont {A.~Z.}\ \bibnamefont
  {Goldberg}}, \bibinfo {author} {\bibfnamefont {L.~L.}\ \bibnamefont
  {S\'anchez-Soto}},\ and\ \bibinfo {author} {\bibfnamefont {H.}~\bibnamefont
  {Ferretti}},\ }\bibfield  {title} {\bibinfo {title} {Intrinsic sensitivity
  limits for multiparameter quantum metrology},\ }\href@noop {} {\bibfield
  {journal} {\bibinfo  {journal} {Phys. Rev. Lett.}\ }\textbf {\bibinfo
  {volume} {127}},\ \bibinfo {pages} {110501} (\bibinfo {year}
  {2021})}\BibitemShut {NoStop}%
\bibitem [{\citenamefont {Chen}\ \emph {et~al.}(2022)\citenamefont {Chen},
  \citenamefont {Chen},\ and\ \citenamefont {Yuan}}]{ChenCY22}%
  \BibitemOpen
  \bibfield  {author} {\bibinfo {author} {\bibfnamefont {H.}~\bibnamefont
  {Chen}}, \bibinfo {author} {\bibfnamefont {Y.}~\bibnamefont {Chen}},\ and\
  \bibinfo {author} {\bibfnamefont {H.}~\bibnamefont {Yuan}},\ }\bibfield
  {title} {\bibinfo {title} {Information geometry under hierarchical quantum
  measurement},\ }\href@noop {} {\bibfield  {journal} {\bibinfo  {journal}
  {Phys. Rev. Lett.}\ }\textbf {\bibinfo {volume} {128}},\ \bibinfo {pages}
  {250502} (\bibinfo {year} {2022})}\BibitemShut {NoStop}%
\bibitem [{\citenamefont {Albarelli}\ and\ \citenamefont
  {Demkowicz-Dobrza\ifmmode~\acute{n}\else \'{n}\fi{}ski}(2022)}]{AlbaD22}%
  \BibitemOpen
  \bibfield  {author} {\bibinfo {author} {\bibfnamefont {F.}~\bibnamefont
  {Albarelli}}\ and\ \bibinfo {author} {\bibfnamefont {R.}~\bibnamefont
  {Demkowicz-Dobrza\ifmmode~\acute{n}\else \'{n}\fi{}ski}},\ }\bibfield
  {title} {\bibinfo {title} {Probe incompatibility in multiparameter noisy
  quantum metrology},\ }\href@noop {} {\bibfield  {journal} {\bibinfo
  {journal} {Phys. Rev. X}\ }\textbf {\bibinfo {volume} {12}},\ \bibinfo
  {pages} {011039} (\bibinfo {year} {2022})}\BibitemShut {NoStop}%
\bibitem [{\citenamefont {Li}\ and\ \citenamefont {Luo}(2016)}]{LiL16}%
  \BibitemOpen
  \bibfield  {author} {\bibinfo {author} {\bibfnamefont {N.}~\bibnamefont
  {Li}}\ and\ \bibinfo {author} {\bibfnamefont {S.}~\bibnamefont {Luo}},\
  }\bibfield  {title} {\bibinfo {title} {{Fisher} concord: Efficiency of
  quantum measurement},\ }\href@noop {} {\bibfield  {journal} {\bibinfo
  {journal} {Quantum Measurements and Quantum Metrology}\ }\textbf {\bibinfo
  {volume} {3}},\ \bibinfo {pages} {44} (\bibinfo {year} {2016})}\BibitemShut
  {NoStop}%
\bibitem [{\citenamefont {Petz}(1996)}]{Petz96}%
  \BibitemOpen
  \bibfield  {author} {\bibinfo {author} {\bibfnamefont {D.}~\bibnamefont
  {Petz}},\ }\bibfield  {title} {\bibinfo {title} {Monotone metrics on matrix
  spaces},\ }\href@noop {} {\bibfield  {journal} {\bibinfo  {journal} {Linear
  Algebra Appl.}\ }\textbf {\bibinfo {volume} {244}},\ \bibinfo {pages} {81}
  (\bibinfo {year} {1996})}\BibitemShut {NoStop}%
\bibitem [{\citenamefont {Petz}\ and\ \citenamefont {Sud\'ar}(1996)}]{PetzS96}%
  \BibitemOpen
  \bibfield  {author} {\bibinfo {author} {\bibfnamefont {D.}~\bibnamefont
  {Petz}}\ and\ \bibinfo {author} {\bibfnamefont {C.}~\bibnamefont {Sud\'ar}},\
  }\bibfield  {title} {\bibinfo {title} {Geometries of quantum states},\
  }\href@noop {} {\bibfield  {journal} {\bibinfo  {journal} {J. Math. Phys.}\
  }\textbf {\bibinfo {volume} {37}},\ \bibinfo {pages} {2662} (\bibinfo {year}
  {1996})}\BibitemShut {NoStop}%
\bibitem [{\citenamefont {Bhatia}(1997)}]{Bhat97book}%
  \BibitemOpen
  \bibfield  {author} {\bibinfo {author} {\bibfnamefont {R.}~\bibnamefont
  {Bhatia}},\ }\href@noop {} {\emph {\bibinfo {title} {Matrix Analysis}}}\
  (\bibinfo  {publisher} {Springer},\ \bibinfo {address} {New York},\ \bibinfo
  {year} {1997})\BibitemShut {NoStop}%
\bibitem [{\citenamefont {Marshall}\ \emph {et~al.}(2011)\citenamefont
  {Marshall}, \citenamefont {Olkin},\ and\ \citenamefont
  {Arnold}}]{MarsOA11book}%
  \BibitemOpen
  \bibfield  {author} {\bibinfo {author} {\bibfnamefont {A.~W.}\ \bibnamefont
  {Marshall}}, \bibinfo {author} {\bibfnamefont {I.}~\bibnamefont {Olkin}},\
  and\ \bibinfo {author} {\bibfnamefont {B.~C.}\ \bibnamefont {Arnold}},\
  }\href@noop {} {\emph {\bibinfo {title} {Inequalities: Theory of Majorization
  and its Applications}}},\ \bibinfo {edition} {2nd}\ ed.,\ Springer Series in
  Statistics\ (\bibinfo  {publisher} {Springer},\ \bibinfo {address} {New
  York},\ \bibinfo {year} {2011})\BibitemShut {NoStop}%
\bibitem [{\citenamefont {Zauner}(2011)}]{Zaun11}%
  \BibitemOpen
  \bibfield  {author} {\bibinfo {author} {\bibfnamefont {G.}~\bibnamefont
  {Zauner}},\ }\bibfield  {title} {\bibinfo {title} {Quantum designs:
  Foundations of a noncommutative design theory},\ }\href@noop {} {\bibfield
  {journal} {\bibinfo  {journal} {Int. J. Quantum Inf.}\ }\textbf {\bibinfo
  {volume} {09}},\ \bibinfo {pages} {445} (\bibinfo {year} {2011})}\BibitemShut
  {NoStop}%
\bibitem [{\citenamefont {Renes}\ \emph {et~al.}(2004)\citenamefont {Renes},
  \citenamefont {Blume-Kohout}, \citenamefont {Scott},\ and\ \citenamefont
  {Caves}}]{ReneBSC04}%
  \BibitemOpen
  \bibfield  {author} {\bibinfo {author} {\bibfnamefont {J.~M.}\ \bibnamefont
  {Renes}}, \bibinfo {author} {\bibfnamefont {R.}~\bibnamefont {Blume-Kohout}},
  \bibinfo {author} {\bibfnamefont {A.~J.}\ \bibnamefont {Scott}},\ and\
  \bibinfo {author} {\bibfnamefont {C.~M.}\ \bibnamefont {Caves}},\ }\bibfield
  {title} {\bibinfo {title} {Symmetric informationally complete quantum
  measurements},\ }\href@noop {} {\bibfield  {journal} {\bibinfo  {journal} {J.
  Math. Phys.}\ }\textbf {\bibinfo {volume} {45}},\ \bibinfo {pages} {2171}
  (\bibinfo {year} {2004})}\BibitemShut {NoStop}%
\bibitem [{\citenamefont {Fuchs}\ \emph {et~al.}(2017)\citenamefont {Fuchs},
  \citenamefont {Hoang},\ and\ \citenamefont {Stacey}}]{FuchHS17}%
  \BibitemOpen
  \bibfield  {author} {\bibinfo {author} {\bibfnamefont {C.~A.}\ \bibnamefont
  {Fuchs}}, \bibinfo {author} {\bibfnamefont {M.~C.}\ \bibnamefont {Hoang}},\
  and\ \bibinfo {author} {\bibfnamefont {B.~C.}\ \bibnamefont {Stacey}},\
  }\bibfield  {title} {\bibinfo {title} {The {SIC} question: History and state
  of play},\ }\href@noop {} {\bibfield  {journal} {\bibinfo  {journal}
  {Axioms}\ }\textbf {\bibinfo {volume} {6}},\ \bibinfo {pages} {21} (\bibinfo
  {year} {2017})}\BibitemShut {NoStop}%
\bibitem [{\citenamefont {Zhu}(2014)}]{Zhu14IOC}%
  \BibitemOpen
  \bibfield  {author} {\bibinfo {author} {\bibfnamefont {H.}~\bibnamefont
  {Zhu}},\ }\bibfield  {title} {\bibinfo {title} {Quantum state estimation with
  informationally overcomplete measurements},\ }\href@noop {} {\bibfield
  {journal} {\bibinfo  {journal} {Phys. Rev. A}\ }\textbf {\bibinfo {volume}
  {90}},\ \bibinfo {pages} {012115} (\bibinfo {year} {2014})}\BibitemShut
  {NoStop}%
\bibitem [{\citenamefont {Bernstein}(2009)}]{Bern09book}%
  \BibitemOpen
  \bibfield  {author} {\bibinfo {author} {\bibfnamefont {D.~S.}\ \bibnamefont
  {Bernstein}},\ }\href@noop {} {\emph {\bibinfo {title} {Matrix Mathematics:
  Theory, Facts, and Formulas}}},\ \bibinfo {edition} {2nd}\ ed.\ (\bibinfo
  {publisher} {Princeton University Press},\ \bibinfo {address} {Princeton,
  NJ},\ \bibinfo {year} {2009})\BibitemShut {NoStop}%
\end{thebibliography}%

\clearpage

\counterwithout{equation}{section}

\setcounter{equation}{0}
\setcounter{figure}{0}
\setcounter{table}{0}
\setcounter{thm}{0}
\setcounter{lem}{0}
\setcounter{remark}{0}
\setcounter{proposition}{0}
\setcounter{corollary}{0}
\setcounter{section}{0}


\renewcommand{\theequation}{S\arabic{equation}}
\renewcommand{\thefigure}{S\arabic{figure}}
\renewcommand{\thetable}{S\arabic{table}}
\renewcommand{\thethm}{S\arabic{theorem}}
\renewcommand{\thelem}{S\arabic{lem}}
\renewcommand{\theremark}{S\arabic{remark}}
\renewcommand{\theproposition}{S\arabic{proposition}}
\renewcommand{\thecorollary}{S\arabic{corollary}}

\renewcommand{\thesection}{S\arabic{section}}
\renewcommand\thesubsection{S\arabic{section}.\arabic{subsection}}


\onecolumngrid
\begin{center}
	\textbf{\large Information Theoretic Significance of Projective Measurements:\\[0.5ex] Supplemental Material}\\
	\vspace{2ex}
	Huangjun Zhu
\end{center}
\bigskip

\twocolumngrid


In this Supplemental Material we prove the key results presented in the main text and appendices, including Propositions~1 and 2, Lemmas 1-8, and Theorems \ref{thm:SharpProjIndex2} and \ref{thm:FisherSharp2}. Incidentally, Theorems \ref{thm:SharpProjIndex} and \ref{thm:FisherSharp} are simple corollaries of Theorems~\ref{thm:SharpProjIndex2} and \ref{thm:FisherSharp2}, respectively.

\section{Proof of \pref{pro:P(A)}}

\begin{proof}[Proof of \pref{pro:P(A)}]
	By definition $\scrP(\scrA)$ is a coarse graining of $\scrA$, so any grouping of $\scrP(\scrA)$ is a coarse graining of $\scrA$. 
	
	To prove the converse, suppose $\scrA=\{A_j\}_{j=1}^m$  contains no zero POVM elements without loss of generality, and $\scrB=\{B_k\}_{k=1}^n$ is a PVM that is a coarse graining of $\scrA$. Then $B_k$ can be expressed as 
	\begin{align}
	B_k=\sum_j \Lambda_{kj} A_j,
	\end{align}
	where $\Lambda$ is a stochastic matrix. Note that $B_k B_{k'}=0$ whenever $k\neq k'$, so all nonzero entries of $\Lambda$ are equal to 1. In other words, $\scrB$ is a grouping of $\scrA$.
	Moreover, $\Lambda_{kj}=1$ iff $\Lambda_{kj'}=1$
	when $A_j, A_{j'}\in \scrA$ belong to a same irreducible component. So $\scrB$ is a grouping of $\scrP(\scrA)$. 
	
	The equality $\gamma(\scrA)=\gamma(\scrP(\scrA))$ follows from the definition of $\gamma(\scrP(\scrA))$. The inequality $\gamma(\scrP(\scrA))\leq d$	follows from the fact that $\gamma(\scrP(\scrA))$
	is equal to
	the number of nonzero POVM elements in the simple PVM $\scrP(\scrA)$. If $\scrA$ is equivalent to a rank-1 PVM, then $\scrP(\scrA)$ is a rank-1 PVM, so the inequality $\gamma(\scrP(\scrA))\leq d$ is saturated. Conversely, if the inequality is saturated, then $\scrP(\scrA)$ is a rank-1 PVM, so $\scrA$ is equivalent to a rank-1 PVM.
\end{proof}

\section{Proofs of \lsref{lem:FisherRank}-\ref{lem:FFbar} and \pref{pro:ProjEqui}}

\begin{proof}[Proof of \lref{lem:FisherRank}]
	The first equality in \eref{eq:FisherRank} follows from the definition in \eref{eq:FIMma}. 		
	The second equality in \eref{eq:FisherRank} follows from the fact that $I(\theta,\scrA)$ is the matrix representation of $\bar{\caF}(\rho,\scrA)$ in the operator basis  $\{\rho_{,a}\}_{a=1}^g$ as shown in \eref{eq:IFbar}. The  third and fourth equalities in \eref{eq:FisherRank} follow from the equation below,
	\begin{align}
	\rk \bar{\caF}(\rho,\scrA)&=\dim \spa (\bar{\scrA})=\dim \spa (\scrA)-1, \nonumber\\
	&=\rk \caF(\rho,\scrA)-1,
	\end{align}
	where $\bar{\scrA}=\{\bar{A}|A\in \scrA\}$. 
\end{proof}

\begin{proof}[Proof of \lref{lem:FIequi}]The equivalence of the
	second and third equalities in \eref{eq:FIequi} follow from \eref{eq:IFbar}. The equivalence of the third and fourth equalities follow from the definition in \eref{eq:FIMma}.

	In addition, the second equality in \eref{eq:FIequi} follows from the first equality by the definition in \eref{eq:FbarDef}. To complete the proof of \lref{lem:FIequi}, it remains to establish the opposite implication.

	Suppose the second equality in \eref{eq:FIequi} holds, that is,	$\bar{\caF}(\rho,\scrA)= \bar{\caF}(\rho,\scrB)$; then $\bar{\caF}(\rho,\scrA)$ and $\bar{\caF}(\rho,\scrB)$ have the same support, which means $\spa(\bar{\scrA})=\spa(\bar{\scrB})$, 
	where
	\begin{align}
	\bar{\scrA}=\{\bar{A}|A\in \scrA\},\quad \bar{\scrB}=\{\bar{B}|B\in \scrB\}. 
	\end{align}
	So  $\spa(\scrA)=\spa(\scrB)$ and $\rho(\scrA)=\rho(\scrB)$. By virtue of \lref{lem:FFbar} in \aref{app:FisherAux} we can now deduce that
	\begin{align}
	&\caF(\rho,\scrA)=[\bar{\caF}^+(\rho,\scrA) + \douter{\rho(\scrA)}{\rho(\scrA)}\lsp]^+\nonumber\\
	&=[\bar{\caF}^+(\rho,\scrB) + \douter{\rho(\scrB)}{\rho(\scrB)}\lsp]^+=\caF(\rho,\scrB), 
	\end{align}
	which confirms the first equality in \eref{eq:FIequi} and completes the proof of \lref{lem:FIequi}. 
\end{proof}

\begin{proof}[Proof of \lref{lem:FFbar}]
	Define the traceless reconstruction subspace of $\scrA$ as the subspace spanned by 
	\begin{align}
	\bar{\scrA}:=\{\bar{A}|A\in \scrA\}, 
	\end{align}
	which coincides with the support of $\bar{\caF}(\rho,\scrA)$. 
	Let $\bar{\Pi}(\scrA)$ be the projector onto this subspace; then 
	\begin{align}
	\bar{\Pi}(\scrA)=\bid\lsp \Pi(\scrA) \bid=\bid\lsp \Pi(\scrA)=\Pi(\scrA) \bid.
	\end{align}

	By definition we have
	\begin{gather}
	\caF^+(\rho,\scrA)\caF(\rho,\scrA)=\caF(\rho,\scrA)\caF^+(\rho,\scrA)=\Pi(\scrA),\\
	\caF(\rho,\scrA)\dket{\rho}=\dket{1}, \quad 
	\caF^+(\rho,\scrA)\dket{1}=\dket{\rho(\scrA)},
	\end{gather}	
	which implies that	
	\begin{align}
	&[\caF^+(\rho,\scrA)-\douter{\rho(\scrA)}{\rho(\scrA)}\lsp]\dket{1}\nonumber\\
	&=\dket{\rho(\scrA)}-\dket{\rho(\scrA)}=0
	\end{align}
	since 	$\tr[\rho(\scrA)]=\tr \rho =1$. So $\caF^+(\rho,\scrA)-\douter{\rho(\scrA)}{\rho(\scrA)}$ is supported in the traceless reconstruction subspace. Consequently, 
	\begin{align}
	&[\caF^+(\rho,\scrA)-\douter{\rho(\scrA)}{\rho(\scrA)}\lsp]\bar{\caF}(\rho,\scrA)
	\nonumber\\
	&=[\caF^+(\rho,\scrA)-\douter{\rho(\scrA)}{\rho(\scrA)}\lsp]\bid\caF(\rho,\scrA) \bid\nonumber\\
	&=[\caF^+(\rho,\scrA)-\douter{\rho(\scrA)}{\rho(\scrA)}\lsp]\caF(\rho,\scrA) \bid\nonumber\\
	&=[\lsp \Pi(\scrA)-\douter{\rho(\scrA)}{1}\lsp] \bid=
	\bar{\Pi}(\scrA),
	\end{align}
	which implies \eref{eq:FFbar} and confirms  \lref{lem:FFbar}. 
\end{proof}

\begin{proof}[Proof of \pref{pro:ProjEqui}]
	The two implications $1\imply 2$ and $1\imply 3$ are obvious; the equivalence of statement 2 and statement  3 follows from \lref{lem:FIequi}. To complete the proof, it remains to prove the implication $3\imply 1$.

	By assumption  $\scrA, \scrB$ are PVMs that contain no zero operators, so each POVM element  $A\in \scrA$ is an eigenvector of $\caF(\rho, \scrA)$, and each $B\in \scrB$ is an eigenvector of $\caF(\rho, \scrB)$ according to the definition in \eref{eq:FDef}. If $\caF(\rho, \scrA)=\caF(\rho, \scrB)$, then each $A\in \scrA$ is an eigenvector of $\caF(\rho, \scrB)$. In conjunction with the following equation
	\begin{align}
	\caF(\rho, \scrB)\dket{A}=\sum_{B\in \scrB} \frac{\tr(A B)\dket{B}}{\tr(\rho B)}
	\end{align}
	we can deduce that each $B\in \scrB$ is either orthogonal to $A$ or supported in the support of $A$. By the same token, each $A\in\scrA$ is either orthogonal to $B$ or supported in the support of $B$. Since 
	$\scrA,\scrB$ are PVMs that contain no zero operators, it follows that $|\scrA|=|\scrB|$, and the  operators in $\scrB$ are identical to the operators in $\scrA$ up to a permutation.	
\end{proof}

\section{Proofs of \lsref{lem:JJbar} and \ref{lem:GMgen}}

\begin{proof}[Proof of \lref{lem:JJbar}]
	Note that $\bar{\caJ}_\rmS(\rho)$ and $\bid$ have the same support: the space of  traceless operators on $\caH$. In addition, the definition in
	\eref{eq:LRS} implies that
	\begin{align}
	[\caS(\rho)-\douter{\rho}{\rho}\lsp]\dket{1}=\dket{\rho}-\dket{\rho}=0,
	\end{align}
	so $\caS(\rho)-\douter{\rho}{\rho}$ is also supported in the space of  traceless operators. Consequently, 
	\begin{align}
	&[\caS(\rho)-\douter{\rho}{\rho}\lsp]\bar{\caJ}_\rmS(\rho)=[\caS(\rho)-\douter{\rho}{\rho}\lsp]\bar{\mathbf{I}}\caJ_\rmS(\rho) \bar{\mathbf{I}}\nonumber\\
	&=[\caS(\rho)-\douter{\rho}{\rho}\lsp]\caJ_\rmS(\rho) \bar{\mathbf{I}}= [\mathbf{I}-\douter{\rho}{1}\lsp] \bar{\mathbf{I}}=
	\bar{\mathbf{I}},
	\end{align}
	which implies the first equality in \eref{eq:JJbar}. The second and third equalities in \eref{eq:JJbar}  can be proved in a similar way. 
\end{proof}

\begin{proof}[Proof of \lref{lem:GMgen}]
	Suppose $\scrA=\{A_j\}_{j=1}^m$ and assume that $A_j\neq 0$ for each $j$ 	without loss of generality. The first equality  in \eref{eq:GMJF} follows from
	the fact that 	$I_\rmS(\theta,\scrA)$ is a matrix representation of $\bar{\caF}_\rmS(\rho,\scrA)$ and has the same nonzero eigenvalues, including multiplicities. The second equality  in \eref{eq:GMJF} follows from the definition in \eref{eq:FFbarS} and the facts that $\bar{\caJ}_\rmS=\bid \caJ_\rmS\bid$ 
	and $\bar{\caF}(\scrA)=\bid \caF(\scrA)\bid$. 
	The third equality  in \eref{eq:GMJF} follows from \essref{eq:JLRS}{eq:FFbarS}{eq:JJbar}  given  that
	\begin{align}
	\dbra{\rho}\caF(\scrA)\dket{\rho}=\sum_j \tr(\rho A_j)=\tr(\rho)=1. 
	\end{align} 
	
	The fourth  equality  in \eref{eq:GMJF} can be proved as follows,	
	\begin{align}
	&\tr \caF_\rmS=\tr\bigl(\caJ_\rmS^{-1/2}\caF\caJ_\rmS^{-1/2}\bigr)=\tr(\caJ_\rmS^{-1}\caF)=\tr(\caS\caF)\nonumber\\
	&
	=\tr\Biggl[\sum_{j=1}^m \frac{\caS\dket{A_j}\dbra{A_j}  }{\tr(\rho A_j)}\Biggr]=\frac{1}{2}\tr\Biggl[\sum_{j=1}^m \frac{\dket{\rho A_j+A_j\rho}\dbra{A_j}  }{\tr(\rho A_j)}\Biggr]\nonumber\\
	&=\sum_{j=1}^m  \frac{\tr(\rho A_j^2) }{\tr(\rho A_j)}.
	\end{align}	
	
	The inequality in \eref{eq:GMJF} can be proved as follows,
	\begin{align}
	&\sum_{j=1}^m  \frac{\tr(\rho A_j^2) }{\tr(\rho A_j)}\leq \sum_{j=1}^m  \frac{\tr(\rho A_j) }{\tr(\rho A_j)}=m,  \label{eq:GMproof1}\\
	&\sum_{j=1}^m \frac{\tr(\rho A_j^2 )}{\tr(\rho A_j)}\leq\sum_{j=1}^m \|A_j\|\leq \sum_{j=1}^m \tr(A_j)=d. \label{eq:GMproof2}
	\end{align}
	Here the inequality in \eref{eq:GMproof1}	
	follows from the fact that $A_j^2\leq A_j$ and is saturated iff each $A_j$ is a projector. The first inequality in \eref{eq:GMproof2} follows from the fact that	$A_j^2\leq \|A_j\| A_j$, and the second inequality is trivial; the two inequalities are saturated simultaneously iff  each POVM element $A_j$ is rank 1.  Therefore, when $m\leq d$, the inequality in \eref{eq:GMJF} is saturated iff $\scrA$ is a PVM with $m$ nonzero POVM elements; when $d\leq m$, the inequality is saturated iff $\scrA$ is rank 1.

	If $I_\rmS$ is replaced by $I_\rmL$ or $I_\rmR$, and $\bar{\caF}_\rmS$, $\caF_\rmS$, $\bar{\caJ}_\rmS$ are replaced accordingly,	
	the same conclusions follow from a similar reasoning. 
\end{proof}

\section{Proof of \lref{lem:FSLRmajor}}

\begin{proof}[Proof of \lref{lem:FSLRmajor}]
	Let $\scrA=\{A_{j}\}_{j=1}^m$ and
	define $m\times m$ Gram matrices  $G_\rmS$,  $G_{\rmL}$, $G_{\rmR}$ with entries
	\begin{align}
	G_{\rmS,jk}&=\frac{\dbra{A_j}\caS(\rho)\dket{A_k}}{\sqrt{\tr(\rho A_j)\tr(\rho A_k)}},\nonumber \\
	&=\frac{\tr(\rho A_kA_j)+\tr(\rho A_jA_k)}{2\sqrt{\tr(\rho A_j)\tr(\rho A_k)}},\\	
	\!\!	G_{\rmL,jk}&=\frac{\dbra{A_j}\caL(\rho)\dket{A_k}}
	{\sqrt{\tr(\rho A_j)\tr(\rho A_k)}}=\frac{\tr(\rho A_kA_j)}{\sqrt{\tr(\rho A_j)\tr(\rho A_k)}},\\	
	\!\!	G_{\rmR,jk}&=\frac{\dbra{A_j}\caR(\rho)\dket{A_k} }{\sqrt{\tr(\rho A_j)\tr(\rho A_k)}}=\frac{\tr(\rho A_jA_k)}{\sqrt{\tr(\rho A_j)\tr(\rho A_k)}}.
	\end{align}
	Then $\caF_\rmS(\rho,\scrA)$, $\caF_\rmL(\rho,\scrA)$, and $\caF_\rmR(\rho,\scrA)$ have
	the same nonzero eigenvalues (including multiplicities) as $G_\rmS$, $G_\rmL$, and $G_\rmR$, respectively. In addition,
	\begin{align}
	G_\rmR=G_\rmL^\rmT=G_\rmL^*, \quad G_\rmS=\frac{1}{2}(G_\rmL+G_\rmR), 
	\end{align}
	which implies that
	\begin{align}
	G_\rmS\preceq G_\rmL\simeq G_\rmR, 
	\end{align} 
which means  $G_\rmS$ is majorized by $G_\rmL$, and $G_\rmL$ is similar to $G_\rmR$. This equation in turn implies \eref{eq:FSLRmajor}.

	Define $m\times m$ Gram matrices  $\bar{G}_\rmS$,  $\bar{G}_{\rmL}$, $\bar{G}_{\rmR}$ with entries
	\begin{align}
	\bar{G}_{\rmS,jk}&=\frac{\dbra{A_j}(\caS(\rho)-\douter{\rho}{\rho})\dket{A_k}}{\sqrt{\tr(\rho A_j)\tr(\rho A_k)}},\\
	\bar{G}_{\rmL,jk}&=\frac{\dbra{A_j}(\caL(\rho)-\douter{\rho}{\rho})\dket{A_k}}
	{\sqrt{\tr(\rho A_j)\tr(\rho A_k)}},\\	
	\bar{G}_{\rmR,jk}&=\frac{\dbra{A_j}(\caR(\rho)-\douter{\rho}{\rho})\dket{A_k} }{\sqrt{\tr(\rho A_j)\tr(\rho A_k)}}.
	\end{align}
	Then $\bar{\caF}_\rmS(\rho,\scrA)$, $\bar{\caF}_\rmL(\rho,\scrA)$, and $\bar{\caF}_\rmR(\rho,\scrA)$ have
	the same nonzero eigenvalues (including multiplicities) as $\bar{G}_\rmS$, $\bar{G}_\rmL$, and $\bar{G}_\rmR$, respectively. In addition,
	\begin{gather}
	\bar{G}_\rmR=\bar{G}_\rmL^\rmT=\bar{G}_\rmL^*, \quad \bar{G}_\rmS=\frac{1}{2}(\bar{G}_\rmL+\bar{G}_\rmR), \\
	\bar{G}_\rmS\preceq \bar{G}_\rmL\simeq \bar{G}_\rmR,
	\end{gather} 
	which imply \eref{eq:FSLRbarMajor} and completes the proof of \lref{lem:FSLRmajor}. 	
\end{proof}

\section{Proof of \lref{lem:FLRnorm}}

\begin{proof}[Proof of \lref{lem:FLRnorm}]
	Since \eref{eq:QCRF} is a simple corollary of \eref{eq:FLRnorm}, we can focus on \eref{eq:FLRnorm}.

	Suppose $\scrA=\{A_j\}_j$ and assume that $A_j\neq 0$ for each $j$ 	without loss of generality; then 
	\esref{eq:JLRS}{eq:FFbarS} imply that 
	\begin{align}
	\caF_\rmS(\rho,\scrA)&=\sum_j \caS(\rho)^{1/2} \frac{\douter{A_j}{A_j}}{\tr(\rho A_j)}\caS(\rho)^{1/2}, \\
	\caF_\rmL(\rho,\scrA)&=\sum_j \frac{\douter{\rho^{1/2}A_j}{\rho^{1/2}A_j}}{\tr(\rho A_j)},\\
	\caF_\rmR(\rho,\scrA)&=\sum_j \frac{\douter{A_j \rho^{1/2}}{A_j\rho^{1/2}}}{\tr(\rho A_j)}.
	\end{align}
	Let $Q$ be the orthogonal projector onto $\caV$; then	
	\begin{align}
	Q=\sum_j A_j. 
	\end{align}	
	In addition, 
	$\caS(\rho)^{1/2}\dket{Q}$, $\dket{\rho^{1/2}Q}$, and $\dket{Q\rho^{1/2}}$ are eigenvectors of $\caF_\rmS(\rho,\scrA)$, $\caF_\rmL(\rho,\scrA)$, and $\caF_\rmR(\rho,\scrA)$, respectively,  with eigenvalue 1, which implies that
	\begin{align}\label{eq:FLRnormProof}
	\|\caF_\rmS(\rho,\scrA)\|\geq 1,\quad  \|\caF_\rmL(\rho,\scrA)\|\geq 1,\quad \|\caF_\rmR(\rho,\scrA)\|\geq 1.
	\end{align}
	
	Let $C$ be an arbitrary operator on $\caH$. Then
	\begin{align}
	&\dbra{C} \caF_\rmL(\rho,\scrA)\dket{C}=\sum_j \frac{|\dinner{C}{\rho^{1/2}A_j}|^2}{\tr(\rho A_j)}
	\nonumber\\
	&=\sum_j \frac{\Bigl|\tr\Bigl[\bigl(CA_j^{1/2}\bigl)^\dag\rho^{1/2}A_j^{1/2}\Bigr]\Bigr|^2}{\tr(\rho A_j)}
	\nonumber\\
	&\leq \sum_j \frac{\tr(CA_jC^\dag)\tr(\rho A_j)}{\tr(\rho A_j)}\nonumber\\
	&=\tr(CQ C^\dag)\leq \tr(CC^\dag), \label{eq:FLRnormProof2}
	\end{align}	
	which implies that $\|\caF_\rmL(\rho,\scrA)\|\leq 1$. In conjunction with \lref{lem:FSLRmajor}, we can deduce  that 
	\begin{align}
	\|\caF_\rmS(\rho,\scrA)\|\leq \|\caF_\rmR(\rho,\scrA)\|=\|\caF_\rmL(\rho,\scrA)\|\leq 1,
	\end{align}
	which confirms \eref{eq:FLRnorm} given \eref{eq:FLRnormProof}. 
	
	Next, suppose $\scrA$ is irreducible, and $\dket{C}$ is an eigenvector of $\caF_\rmL(\rho,\scrA)$ with eigenvalue 1. Then the two inequalities in \eref{eq:FLRnormProof2} are saturated, which means
	\begin{gather}
	CA_j^{1/2}\propto \rho^{1/2}A_j^{1/2}\quad \forall j=1,2,\ldots m,\\
	\supp(C)\leq \caV. 
	\end{gather}
	Let $B=\rho^{-1/2}C$; then
	\begin{gather}
	BA_j\propto A_j\quad \forall j=1,2,\ldots m,\label{eq:BAj}\\
	\supp(B)\leq \caV,\\
	\range(B)\leq \caV,
	\end{gather}
	given that 	 $\scrA=\{A_j\}_j$ is a POVM on $\caV$. So  
	each $A_j$ is supported in an eigenspace of $B$, which means $B$ is diagonalizable, and the direct sum of its eigenspaces within $\caV$ coincides with $\caV$.  According to \lref{lem:ProjDirectSum} below, POVM elements of $\scrA$ supported in different eigenspaces of $B$ are necessarily orthogonal. Since $\scrA$ is irreducible by assumption, it follows that $B$ has only one eigenspace within $\caV$. In other words, $B$ is proportional to the projector $Q$ onto $\caV$, which means  
	\begin{align}
	C=\rho^{1/2} B\propto \rho^{1/2}Q. 
	\end{align}
	So the eigenspace of $\caF_\rmL(\rho,\scrA)$ associated with the eigenvalue 1 is nondegenerate, and so is the eigenvalue~1 itself, given that
	$\caF_\rmL(\rho,\scrA)$ is Hermitian. 
	
	In conjunction with \lref{lem:FSLRmajor} and \eref{eq:FLRnorm}, we can deduce  that the eigenvalue 1 is nondegenerate for   $\caF_\rmS(\rho,\scrA)$  and $\caF_\rmR(\rho,\scrA)$ as well. 
\end{proof}

Next, we prove an auxiliary lemma employed in the proof of \lref{lem:FLRnorm}.  
\begin{lem}\label{lem:ProjDirectSum}
	Suppose  $\caH$ has direct sum decomposition $\caH=\caH_1\oplus \caH_2\oplus \cdots\oplus\caH_m$ (without assuming orthogonal direct sum) and 
	$\scrA=\{A_{j}\}_{j=1}^m$ is a POVM on $\caH$ such that $A_j$ is supported in $\caH_j$ for $j=1,2,\ldots,m$.  Then
	the direct sum is actually an orthogonal direct sum,	 $\scrA$ is a PVM, and $A_j$ is the orthogonal projector onto $\caH_j$. 
\end{lem}

\begin{proof}[Proof of \lref{lem:ProjDirectSum}]
	Let $\caH_j'$ be the orthogonal complement of the direct sum of $\caH_1,\ldots, \caH_{j-1}, \caH_{j+1},\ldots, \caH_m$ and $P_j$ be the orthogonal projector onto $\caH'_j$. Then we have $\dim\caH_j'=\dim\caH_j$ for each $j$ and 
	\begin{align}
	d=\sum_{j=1}^m\tr(P_j)=\sum_{j=1}^m\tr(P_j A_j)\leq \sum_{j=1}^m \tr(A_j)=d,
	\end{align}
	which implies that $\tr(P_jA_j)=\tr(A_j)$ for each $j$. So $A_j$ is supported in the support of $P_j$, $\tr(A_j)\leq \tr(P_j)$, and $\tr(A_j A_k)=0$ for $j\neq k$. Consequently,
	\begin{align}
	d=\sum_{j=1}^m \tr(A_j)\leq \sum_{j=1}^m \tr(P_j)=d,
	\end{align}
	which implies that $A_j=P_j$ and $\tr(P_j P_k)=0$ for $j\neq k$, so $\scrA$ is a PVM. In addition, 
	\begin{align}
	\caH'_j=\supp(P_j)=\supp(A_j)\leq \caH_j,
	\end{align}	
	which implies that $\caH'_j=\caH_j$ and $\supp(A_j)= \caH_j$ since $\caH'_j$ and $\caH_j$ have the same dimension. So  the direct sum in \lref{lem:ProjDirectSum} is actually an orthogonal direct sum,  and $A_j$ is the orthogonal projector onto $\caH_j$ for $j=1, 2,\ldots, m$. 	
\end{proof}

\section{Proof of \lref{lem:FisherEigMul}}

\begin{proof}[Proof of \lref{lem:FisherEigMul}]

	Suppose $A$ and $B$ are two operators on a finite-dimensional Hilbert space; then $AB$  and $BA$ have the same eigenvalues, including the algebraic  multiplicities, according to Corollary 4.4.11 in \rcite{Bern09book}. A similar conclusion also holds for superoperators. Based on this observation and \eref{eq:FFbarS} we can deduce that
	\begin{align}
	\mu_\lambda(\caF_\rmS)&= \mu_\lambda(\caS^{1/2}\caF\caS^{1/2})=\mu_\lambda(\caF\caS),  \label{eq:FisherEigMulProof1}\\
	\mu_\lambda(\bar{\caF}_\rmS)&=\mu_\lambda\bigl((\bar{\caJ}_\rmS^+)^{1/2}\caF(\bar{\caJ}_\rmS^+)^{1/2}\bigr)=\mu_\lambda(\caF\bar{\caJ}_\rmS^+)\nonumber\\
	&=\mu_\lambda(\caF(\caS-\douter{\rho}{\rho})), \label{eq:FisherEigMulProof2}
	\end{align}
	where the last equality follows from \lref{lem:JJbar}.
	Incidentally, the algebraic multiplicity of each eigenvalue is larger than or equal to the geometric multiplicity. The three superoperators $\caF_\rmS$, $\caF\caS$, and $\bar{\caF}_\rmS$ are  diagonalizable, so  the algebraic multiplicity of each eigenvalue coincides with the geometric multiplicity.

	On the other hand, 
	\begin{align}
	\rk(\bar{\caF}_\rmS)=\rk(\bar{\caF})=\rk(\caF)-1=\rk(\caF_\rmS)-1, 
	\end{align}
	which means
	\begin{align}
	\mu_0(\bar{\caF}_\rmS)=\mu_0(\caF_\rmS)+1. \label{eq:FisherEigMulProof3}
	\end{align}
	Furthermore, \Lsref{lem:GMgen} and \ref{lem:FLRnorm}
	imply that
	\begin{gather}
	\tr(\bar{\caF}_\rmS)=\tr(\caF_\rmS)-1,\quad 
	\|\bar{\caF}_\rmS\|\leq \|\caF_\rmS\|= 1, \label{eq:FisherEigMulProof4}
	\end{gather}
	so all eigenvalues of $\bar{\caF}_\rmS$ and $\caF_\rmS$ lie in the closed interval $[0,1]$ given that $\bar{\caF}_\rmS$ and $\caF_\rmS$ are positive semidefinite. 
	
	Note that $\rho$ is an eigenvector of $\caS\caF=(\caF\caS)^\dag$ with eigenvalue 1. 
	Suppose $\lambda\neq 1$ and $\dket{C}$ is an eigenvector of $\caF\caS$ with eigenvalue $\lambda$. Then 
	\begin{align}
	\lambda\dinner{\rho}{C}=\dbra{\rho}\caF\caS\dket{C}=\dinner{\rho}{C},
	\end{align}
	which implies that $\dinner{\rho}{C}=0$ and 
	\begin{align}
	\caF(\caS-\douter{\rho}{\rho})\dket{C}=\caF\caS \dket{C}=\lambda \dket{C}. 
	\end{align}
	So any eigenvector of  $\caF\caS$ with eigenvalue $\lambda\neq 1$ is also an eigenvector of  $\caF(\caS-\douter{\rho}{\rho})$ with the same eigenvalue, which means 
	\begin{align}
	\mu_\lambda(\caF_\rmS)=	\mu_\lambda(\caF\caS)\leq \mu_\lambda (\caF(\caS-\douter{\rho}{\rho}))
	=
	\mu_\lambda(\bar{\caF}_\rmS)
	\end{align}
	for $\lambda \neq 1$	in view of \esref{eq:FisherEigMulProof1}{eq:FisherEigMulProof2}. Together  with \esref{eq:FisherEigMulProof3}{eq:FisherEigMulProof4}, this equation implies \eref{eq:FisherEigMul}. 
	
	When $\bar{\caF}_\rmS$ is replaced by $\bar{\caF}_\rmL$ or $\bar{\caF}_\rmR$, and
	$\caF_\rmS$ is replaced by $\caF_\rmL$ or $\caF_\rmR$  accordingly, the same conclusions follow from a similar reasoning. 
\end{proof}

\section{Proofs of \thsref{thm:SharpProjIndex2} and \ref{thm:FisherSharp2}}

\begin{proof}[Proof of \thref{thm:SharpProjIndex2}]
	Let  $\scrA_\alpha$ for $\alpha=1,2,\ldots, \gamma(\scrA)$ be the irreducible components in $\scrA$. Then $\caF_\rmS(\rho,\scrA_\alpha)$ are mutually orthogonal and
	\begin{align}
	\caF_\rmS(\rho,\scrA)=\sum_{\alpha=1}^{\gamma(\scrA)} \caF_\rmS(\rho,\scrA_\alpha).
	\end{align}
	Therefore,
	\begin{align}
	\mu(\caF_\rmS(\rho,\scrA))=\sum_{\alpha=1}^{\gamma(\scrA)} \mu(\caF_\rmS(\rho,\scrA_\alpha))=\gamma(\scrA),
	\end{align}	
	which confirms \eref{eq:SharpProjIndex21}.	
	Here the second equality follows from the fact that $\mu(\caF_\rmS(\rho,\scrA_\alpha))=1$ by  \lref{lem:FLRnorm}. 
	
	The first equality in \eref{eq:SharpProjIndex22} holds by definition. The second equality in \eref{eq:SharpProjIndex22} follows from the fact that  $I_\rmS(\theta,\scrA)$ and $\bar{\caF}_\rmS(\rho,\scrA)$ have the same nonzero spectrum (including multiplicities). The third equality in \eref{eq:SharpProjIndex22} follows from \esref{eq:SharpProjIndex21}{eq:FisherEigMul}, given that $\mu(\cdot)$ is a shorthand for  $\mu_1(\cdot)$.

	By virtue of \lref{lem:FLRnorm}  we can also deduce that
	\begin{align}
	\mu(\caF_\rmL(\rho,\scrA))&=\sum_{\alpha=1}^{\gamma(\scrA)} \mu(\caF_\rmL(\rho,\scrA_\alpha))=\gamma(\scrA),\\
	\mu(\caF_\rmR(\rho,\scrA))&=\sum_{\alpha=1}^{\gamma(\scrA)} \mu(\caF_\rmR(\rho,\scrA_\alpha))=\gamma(\scrA). 
	\end{align}	
	Together with \lref{lem:FisherEigMul} and the discussion in \aref{app:RLD}, the two equations imply that	
	\begin{align}
	\mu(I_\rmL(\theta,\scrA))&=\mu(I_\rmR(\theta,\scrA))=	\mu(\bar{\caF}_\rmL(\rho,\scrA))\nonumber\\
	&=	\mu(\bar{\caF}_\rmR(\rho,\scrA))=\gamma(\scrA)-1.
	\end{align}	
	So \esref{eq:SharpProjIndex21}{eq:SharpProjIndex22}	 
	still hold if $\caF_\rmS$ is replaced by $\caF_\rmL$ or $\caF_\rmR$, and $\bar{\caF}_\rmS$ ($I_\rmS$) is replaced by $\bar{\caF}_\rmL$ ($I_\rmR$) or $\bar{\caF}_\rmR$ ($I_\rmL$) accordingly.

	Let $Q_\alpha$ be the orthogonal projectors associated with the irreducible components $\scrA_\alpha$. 
	Then it is easy to verify that 
	$\caS(\rho)^{1/2}\dket{Q_\alpha}$ for $\alpha=1,2,\ldots, \gamma(\scrA)$ are mutually orthogonal eigenvectors of $\caF_\rmS(\rho,\scrA)$ with eigenvalue~1. Similarly,  $\dket{\rho^{1/2}Q_\alpha}$ are mutually orthogonal eigenvectors of $\caF_\rmL(\rho,\scrA)$ with eigenvalue 1, and $\dket{Q_\alpha \rho^{1/2}}$ are mutually orthogonal eigenvectors of $\caF_\rmR(\rho,\scrA)$ with eigenvalue~1. All these observations are consistent with \eref{eq:SharpProjIndex21}. 
\end{proof}

\begin{proof}[Proof of \thref{thm:FisherSharp2}]
	To achieve our goal, we shall prove
	the implications $1\imply 2$, $2\imply 4$, $4\imply 5$, $5\imply 3$, and $3\imply 1$ successively.

	By assumption $\scrA$ contains no zero POVM elements. If $\scrA$ is a PVM, then 
	\begin{align}
	\mu(I_\rmS(\theta,\scrA))=\tr I_\rmS(\theta,\scrA)=|\scrA|-1\quad \forall \rho(\theta)\in \scrD(\caH)
	\end{align}
	according to \lref{lem:GMgen} and \thref{thm:SharpProjIndex2}, 
	which implies that  $I_\rmS(\theta,\scrA)$ is a projector  for any $\rho(\theta)\in \scrD(\caH)$, given that the eigenvalues of $I_\rmS(\theta,\scrA)$ lie in the closed interval $[0,1]$ by \eref{eq:FIMmaQCRGMI}. So $\scrA$ is  universally Fisher sharp, which confirms the implication $1\imply 2$.

	If 	$\scrA$ is universally Fisher sharp, then $I_\rmS(\theta,\scrA)$ is a projector  for each $\rho(\theta)\in \scrD(\caH)$ by definition, so all  operators in statement 4 are projectors for each $\rho(\theta)\in \scrD(\caH)$ according to \lref{lem:FIprojectors} below. This observation confirms the implication $2\imply 4$. 
	
	The implication $4\imply 5$ is obvious. The implication $5\imply 3$ follows from \lref{lem:FIprojectors} below.
	
	It remains to prove the implication $3\imply 1$. Suppose  $\scrA$ is Fisher sharp at $\rho=\rho(\theta)$; then $I_\rmS(\theta,\scrA)$ is a projector at $\rho(\theta)$, so  $\caF_\rmS(\rho,\scrA)$ is also a projector at $\rho$ by \lref{lem:FIprojectors} below.	Therefore,
	\begin{align}
	\gamma(\scrA)&=\mu(\caF_\rmS(\rho,\scrA))=\rk(\caF_\rmS(\rho,\scrA))\nonumber\\
	&=\rk(\caF(\rho,\scrA))=\dim(\spa(\scrA)),
	\end{align}
	where the first equality follows from \thref{thm:SharpProjIndex2}, the second inequality follows from the fact that $\caF_\rmS(\rho,\scrA)$ is a projector,  the third equality follows from  the definition in \eref{eq:FFbarS}, and the fourth equality follows from \lref{lem:FisherRank}.
	The above equation implies that all POVM elements in any given irreducible component of $\scrA$ are proportional to a positive operator. 
	So $\scrA$ is a PVM, given that $\scrA$ is a simple POVM by assumption. This observation confirms the implication $3\imply 1$ and completes the proof of \thref{thm:FisherSharp2}. 
\end{proof}

\begin{lem}\label{lem:FIprojectors}
	Suppose $\rho(\theta)\in \scrD(\caH)$ and $\scrA$ is a  POVM on $\caH$. If one of the nine operators/superoperators $I_\rmS$, $I_\rmL$, $I_\rmR$, $\bar{\caF}_\rmS$, $\bar{\caF}_\rmL$, $\bar{\caF}_\rmR$, $\caF_\rmS$, $\caF_\rmL$, $\caF_\rmR$ is a projector at $\rho(\theta)$, then all of them are projectors  at $\rho(\theta)$. 
\end{lem}
\begin{proof}[Proof of \lref{lem:FIprojectors}]
	Note that $I_\rmS$ and $\bar{\caF}_\rmS$ have the same nonzero spectrum; $I_\rmL$ and $\bar{\caF}_\rmR$ have the same nonzero spectrum;  $I_\rmR$ and $\bar{\caF}_\rmL$ have the same nonzero spectrum (see \aref{app:RLD}). 
	In addition, 	
	all the operators/superoperators in \lref{lem:FIprojectors} are positive semidefinite and their eigenvalues lie in the closed interval $[0,1]$ according to \eref{eq:FIMmaQCRGMI} and \lref{lem:FLRnorm}. So $I_\rmS$ is a projector iff 
	\begin{align}
	\tr(I_\rmS)=\mu(I_\rmS),
	\end{align}
	and the same conclusion holds for the other eight operators/superoperators.

	Furthermore,  \lref{lem:GMgen} and \thref{thm:SharpProjIndex2} (cf. \lref{lem:FisherEigMul}) imply that
	\begin{align}
	\tr(I_\rmS)&=\tr(I_\rmL)=\tr(I_\rmR)=\tr(\bar{\caF}_\rmS)=\tr(\bar{\caF}_\rmL)=\tr(\bar{\caF}_\rmR)\nonumber\\
	&
	=\tr(\caF_\rmS)-1=\tr(\caF_\rmL)-1=\tr(\caF_\rmR)-1,\\
	\mu(I_\rmS)&=\mu(I_\rmL)=\mu(I_\rmR)=\mu(\bar{\caF}_\rmS)=\mu(\bar{\caF}_\rmL)=\mu(\bar{\caF}_\rmR)\nonumber\\
	&
	=\mu(\caF_\rmS)-1=\mu(\caF_\rmL)-1=\mu(\caF_\rmR)-1. 
	\end{align}
	Therefore, if one of the nine operators/superoperators in \lref{lem:FIprojectors} is a projector, then all of them are projectors, which completes the proof of \lref{lem:FIprojectors}.	
\end{proof}

\end{document}